\documentclass[twocolumn,10pt]{article}

\usepackage{geometry}
\usepackage{framed}

\renewenvironment{abstract}{
  \small
  \list{}{%
    \setlength{\leftmargin}{2cm}
    \setlength{\rightmargin}{\leftmargin}%
  }%
  \item\relax}
  {\endlist
}

\usepackage{cuted}
\newenvironment{widetext}{\begin{strip}}{\end{strip}}

\usepackage{authblk}

\makeatletter
\def\@maketitle{%
  \newpage
  \null
  \vskip 2em%
  \begin{center}%
  \let \footnote \thanks
    {\large\bf \@title \par}
    \vskip 1.5em%
    {\large
      \lineskip .5em%
      \begin{tabular}[t]{c}%
	\parbox{0.7\textwidth}{\centering \normalsize \@author}
      \end{tabular}\par}%
  \end{center}%
  \par
  \vspace{-0.3cm}}
\makeatother

\usepackage[center,small]{titlesec}
\titlelabel{Appendix \thetitle:\quad}

\usepackage{multicol}

\usepackage{siunitx}

\usepackage{amsthm}
\makeatletter
\newtheoremstyle{definition}
  {3pt}
  {3pt}
  {\sl}
  {}
  {\bf}
  {:}
  {.5em}
  {\thmname{#1}\thmnumber{\@ifnotempty{#1}{ }\@upn{#2}}
    \thmnote{ {\bfseries(#3)}}}
\makeatother
\theoremstyle{definition}
\newtheorem{defn}{Definition}

\makeatletter
\newtheoremstyle{convention}
  {3pt}
  {3pt}
  {\sl}
  {}
  {\bf}
  {:}
  {.5em}
  {\thmname{#1}\thmnumber{\@ifnotempty{#1}{ }\@upn{#2}}
    \thmnote{ {\bfseries(#3)}}}
\makeatother
\theoremstyle{convention}
\newtheorem*{conv}{Conventions}

\usepackage{xpatch}
\xpatchbibmacro{volume+number+eid}{\setunit*{\adddot}}{}{}{} 

\usepackage[labelfont=bf]{caption}

\usepackage{graphicx}
\usepackage{subcaption}
\usepackage[labelformat=simple]{subcaption}

\usepackage{amsmath}
\usepackage{amssymb}
\usepackage{mathtools}
\usepackage[usenames,dvipsnames]{xcolor}
\usepackage{hyperref}
\usepackage{enumitem}
\setlist[itemize,1]{label={\raisebox{0.35ex}{\tiny$\bullet$}}}

\usepackage{tikz}
\usetikzlibrary{calc} 
\usetikzlibrary{arrows} 
\usetikzlibrary{decorations} 
\usepackage{pgfplots}

\makeatletter
\newtheoremstyle{prop}
  {3pt}
  {3pt}
  {\sl}
  {}
  {\bf}
  {:}
  {.5em}
  {\thmname{#1}\thmnumber{\@ifnotempty{#1}{ }\@upn{#2}}
    \thmnote{ {\bfseries(#3)}}}
\makeatother
\theoremstyle{prop}

\newtheorem{lem}[defn]{Lemma}
\newtheorem{thm}[defn]{Theorem}
\newtheorem*{thm*}{Theorem}

\newcommand{\Q}{Q}
\newcommand{\Qent}{Q_{\textnormal{ent}}}
\newcommand{\ca}[1]{\mathcal{#1}}
\newcommand{\chan}{\Lambda}

\newcommand{\ket}[1]{|#1\rangle}
\newcommand{\bra}[1]{\langle#1|}
\newcommand{\Hmin}{H_{\textnormal{min}}}
\newcommand{\Hmax}{H_{\textnormal{max}}}
\newcommand{\tr}{\textnormal{tr}}
\newcommand{\xb}{X}

\newcommand{\zb}{Z}
\newcommand{\nb}{D}
\newcommand{\ix}{I_X}
\newcommand{\iz}{I_Z}
\newcommand{\ex}{\gamma}

\newcommand{\ez}{\lambda}
\newcommand{\Ez}{\Lambda}
\newcommand{\nx}{n}
\newcommand{\nz}{k}
\newcommand{\N}{N}
\newcommand{\mux}{\mu_x}
\newcommand{\muz}{\mu_z}

\newcommand{\etolx}{e_x}
\newcommand{\etolz}{e_z}
\newcommand{\exest}{e_x}
\newcommand{\ezest}{e_z}
\newcommand{\ppass}{p_{\textnormal{pass}}}
\newcommand{\pabort}{p}

\newcommand{\ketbra}[1]{\ket{#1}\bra{#1}}
\newcommand{\End}{\textnormal{End}}
\newcommand{\Herm}{\textnormal{Herm}}
\newcommand{\idop}[1]{I_{#1}}
\newcommand{\idchan}[1]{I_{#1}}
\newcommand{\dt}[1]{\textbf{#1}}
\DeclarePairedDelimiter\abs{\lvert}{\rvert}%
\DeclareMathOperator{\Forall}{\forall}
\renewcommand{\epsilon}{\varepsilon}

\DeclarePairedDelimiter\norm{\lVert}{\rVert}%
\makeatletter
\let\oldabs\abs
\def\abs{\@ifstar{\oldabs}{\oldabs*}}
\let\oldnorm\norm
\def\norm{\@ifstar{\oldnorm}{\oldnorm*}}
\makeatother

\begin{document}

\title{Capacity estimation and verification \\ of quantum channels
  with arbitrarily correlated errors}
\author[1,2]{Corsin Pfister}
\author[1,3]{M. Adriaan Rol}
\author[1,4]{Atul Mantri}
\author[2,5]{Marco Tomamichel}
\author[1]{Stephanie Wehner}
\affil[1]{QuTech, Delft University of Technology, Lorentzweg 1, 2628 CJ
  Delft, Netherlands}
\affil[2]{Centre for Quantum Technologies, 3 Science
  Drive 2, Singapore 117543}
\affil[3]{Kavli Institute of Nanoscience, Delft University of Technology, P.O. Box 5046, 2600 GA Delft, The Netherlands}
\affil[4]{Singapore University of Technology and Design, 20 Dover Drive,
  Singapore 138682}
\affil[5]{University of Sydney, School of Physics, NSW 2006 Sydney,
  Australia}

\twocolumn[

\maketitle

\begin{abstract}
  One of the main figures of merit for quantum memories and quantum
  communication devices is their quantum capacity. It has been studied for
  arbitrary kinds of quantum channels, but its practical estimation has so far
  been limited to devices that implement independent and identically distributed
  (i.i.d.) quantum channels, where each qubit is affected by the same noise process.
  Real devices, however, typically exhibit correlated errors.

  Here, we overcome this limitation by presenting protocols that
  estimate a channel's \emph{one-shot} quantum capacity for the case where the
  device acts on (an arbitrary number of) qubits. The one-shot quantum capacity
  quantifies a device's ability to store or communicate quantum information,
  even if there are correlated errors across the different qubits.

  We present a protocol which is easy to implement and which comes in two
  versions. The first version estimates the one-shot quantum capacity by
  preparing and measuring in two different bases, where all involved qubits are
  used as test qubits. The second version verifies on-the-fly that a channel's
  one-shot quantum capacity exceeds a minimal tolerated value while storing or
  communicating data, therefore combining test qubits and data qubits in one
  protocol. We discuss the performance of our method using simple examples, such as the dephasing
  channel for which our method is asymptotically optimal. Finally, we apply our method to
  estimate the one-shot capacity in an experiment using a transmon qubit.
\end{abstract}

\vspace{1cm}
]

\section*{Introduction}

One of the main obstacles on the way to quantum computers and quantum
communication networks is the problem of noise due to imperfections in the
devices. Noise is caused by uncontrolled interactions of the quantum information
carriers with their environment. These interactions take place at all stages:
when the carriers are processed, when they are transmitted and when they are
stored. Physicists and engineers spend large efforts in developing noise
protection measures, and assessing their performance is crucial for the
development of quantum information processing devices. In this article, we focus
on the estimation of noise in the storage and transmission of the quantum
information carriers, that is, we describe methods to assess \emph{quantum
memory} and \emph{quantum communication} devices.

In the language of quantum information theory, memory and communication devices
are described by a \emph{quantum channel}, which is a function $\chan$ that maps
an input state $\rho_{\text{in}}$ of the device to its output state
$\rho_{\text{out}} = \chan(\rho_{\text{in}})$ (see Section~\ref{app:preliminaries} for
a precise definition). In this unified description, assessing the noise in a
quantum device reduces to estimating the decoherence of a quantum channel. One
way to achieve this is through \emph{quantum process tomography} \cite{CN97},
which aims at completely determining the channel from measurement data (see
e.g.~\cite{FR16,FBK16} for more recent works on tomography, and
e.g.~\cite{Gre15,PR04} for surveys on specific types of tomography). This comes
with two major disadvantages. Firstly, process tomography typically only works
for channels that behave the same way in every run of the experiment (formalized
by the \emph{i.i.d.  assumption} - for \emph{independent and identically
  distributed}), or under some symmetry assumptions. This assumption is violated
for many devices that are used in practice, which typically show correlated
errors. Secondly, since process tomography aims at a complete characterization
of the channel, it requires the collection of large amounts of data for many
combinations of input states and measurement settings.  A complete
characterization of a channel is certainly useful (as all properties of the
channel can be inferred from it), but it is very costly if the task at hand is
to simply estimate a figure of merit of the channel. For quantum storage and
quantum communication devices, a central figure of merit is the \emph{quantum
  capacity} of the channel, which quantifies the amount of quantum information
that can be stored or transmitted by the device~\cite{NC00}.  While the
deployment of a suitable error correcting code requires knowledge of the
specifics of the channel, an estimate of the quantum capacity is of great use
when assessing the usefulness of the tested device.

In this work, we present a method to estimate the \emph{one-shot} quantum
capacity $\Q^\epsilon(\chan)$ of a quantum channel $\chan$. While the quantum
capacity $\Q$ only makes statements for devices that behave identically under
many repeated uses, the one-shot quantum capacity $\Q^\epsilon$ applies to the
more general case of devices with arbitrarily correlated errors. It quantifies
the number of qubits that can be sent through the channel with a fidelity of at
least $1-\epsilon$ in a single use of the device using the best possible error correcting code
(we will explain this in more detail in the next section). We present a protocol
that allows to estimate $\Q^\epsilon(\chan)$ from data obtained from simple
measurements. In addition to dealing with arbitrarily correlated errors, it has
the advantage of requiring fewer measurement settings than quantum process
tomography. Our method can also be used to assess whether a possibly
imperfect error-correction scheme forms an improvement. This is the case if the error-corrected channel
has a higher capacity than what we would otherwise expect.

\section*{Results}

\subsection*{The one-shot quantum capacity}

Noise can be modelled as a channel $\chan$, which is given as a map
\begin{align}
  \label{eq:mem}
  \chan: \ca{S}(\ca{H}) \rightarrow \ca{S}(\ca{H}) \,,
\end{align}
where $\ca{S}(\ca{H})$ denotes the set of quantum states on the Hilbert space of the system that is being
stored or transmitted.
For reasons of illustration, we will discuss channels of storage
devices here, but mathematically, nothing is different for communication
devices. In the realm of communication, it is convenient to think of a sender (Alice) who wants
to relay qubits to a receiver (Bob). For memory device, Alice and Bob simply label the input and output.

Consider a quantum memory device designed for storing a quantum system with
Hilbert space $\ca{H}$ for some time interval $\Delta t$. Ideally, it leaves the
state of the system completely invariant over that time span, but real storage
devices are always subject to noise. A measure for how well the channel $\chan$
preserves the state of the system is obtained by minimizing the square of the
\emph{fidelity} between the input state $\ket{\phi}$ and the output state
$\chan(\phi)$,
\begin{align}
  F(\ket{\phi}, \chan(\phi)) = \sqrt{ \bra{\phi} \chan(\phi) \ket{\phi} } \,,
\end{align}
over all possible input states $\ket{\phi} \in \ca{H}$,
\begin{align}
  \label{eq:minf}
  \min_{\ket{\phi} \in \ca{H}} F^2(\ket{\phi}, \chan(\phi)) =
  \min_{\ket{\phi} \in \ca{H}} \bra{\phi} \chan(\phi) \ket{\phi} \,.
\end{align}
Low values of the quantity \eqref{eq:minf} imply that if the device is used
without modification, then at least some states of the system are strongly
affected by the channel, therefore introducing errors. However, this does not
necessarily mean that the device is useless as a storage device, as this
quantity does not account for the possibility that such errors can be corrected
using \emph{quantum error correction (QEC)}.

\begin{figure}[t]
  \begin{center}
  \includegraphics[width=\linewidth]{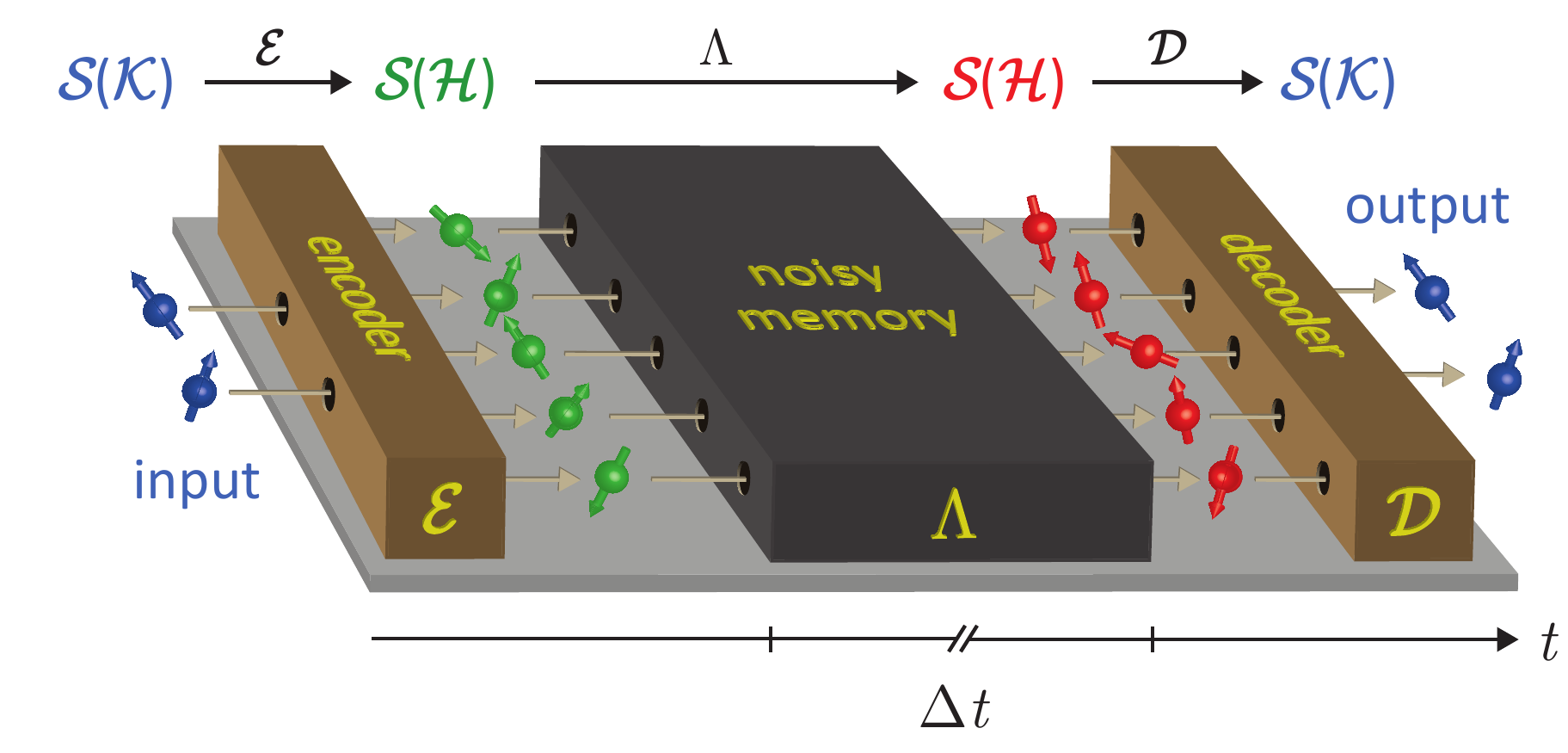}
  \caption{\textbf{Time diagram of an error-corrected quantum memory.}
    An error-correcting code can turn a noisy quantum memory for some system
    with a Hilbert space $\ca{H}$ into an approximately noise-free memory for
    some smaller system with a lower-dimensional Hilbert space $\ca{K}$.
	Such a
    code consists of an \emph{encoder} $\ca{E}$, which is applied before the
    quantum memory, and a \emph{decoder} $\ca{D}$, which is applied after the
    quantum memory. The encoder maps the state space $\ca{K}$ of the smaller
    system into a subspace $\ca{H}' \subseteq \ca{H}$ of the larger system that
    is stored by the quantum memory, so it implements an encoding channel
    $\ca{E}: \ca{S}(\ca{K}) \rightarrow \ca{S}(\ca{H})$. The goal is to design
    the encoder such that the image $\ca{E}(\ca{S}(\ca{K})) = \ca{S}(\ca{H}')
    \subseteq \ca{S}(\ca{H})$ is a subspace that is left approximately intact by
    the quantum memory. Then, the decoder can be chosen such that it implements
    a channel $\ca{D}: \ca{S}(\ca{H}) \rightarrow \ca{S}(\ca{K})$ which maps
    that intact subspace back to the state space of the smaller system.  This
    leads to an error-corrected memory for the smaller system which implements
    the channel $\ca{D} \circ \chan \circ \ca{E}: \ca{S}(\ca{K}) \rightarrow
    \ca{S}(\ca{K})$. Note that this figure shows a time diagram, so the three
    devices are not necessarily placed in the same spatial order as they appear
    in the figure.
    \label{fig:coding}}
  \end{center}
\end{figure}

An error-correcting code for a channel $\chan$ consists of an \emph{encoding}
$\ca{E}$, which is applied before the channel, and a \emph{decoding} $\ca{D}$,
which is applied after the channel (see the explanations in Figure~\ref{fig:coding}).
Together, these devices form an error-corrected quantum memory for a smaller
system, implementing a channel
\begin{align}
  \ca{D} \circ \chan \circ \ca{E}: \ca{S}(\ca{K}) \rightarrow \ca{S}(\ca{K}) \,,
\end{align}
where $\ca{K}$ is the Hilbert space of the smaller system and where $\circ$
denotes the composition of maps. Instead of evaluating the quantity
\eqref{eq:minf} for the channel~$\chan$ directly, it should be evaluated for
such a corrected channel $\ca{D} \circ \chan \circ \ca{E}$. A figure of merit
for the usefulness of the quantum memory is then given by the size of the
largest system $\ca{K}$ that can be stored in the memory using such an
error-correcting code. This is identical to the largest subspace $\ca{H}'
\subseteq \ca{H}$ that is left approximately invariant by the memory. This is
quantified by the \emph{one-shot quantum capacity} $\Q^\epsilon(\chan)$,
defined by \cite{BD10,TBR16}
\begin{align}
  \label{eq:cap}
  \Q^\epsilon(\chan) := \max \{ \log m \mid F_\text{min}(\chan, m) \geq 1 -
  \epsilon \} \,,
\end{align}
where
\begin{align}
  \label{eq:fmin}
  F_\text{min}(\chan, m) :=
  \max_{\ca{H}' \subseteq \ca{H} \atop \dim(\ca{H}') = m} \max_{\ca{D}}
  \min_{\ket{\phi} \in \ca{H}'} \bra{\phi} (\ca{D} \circ \chan)(\phi)
  \ket{\phi}
\end{align}
and where the inner maximum is taken over all possible decoders $\ca{D}: \ca{S}(\ca{H})
\rightarrow \ca{S}(\ca{H})$. The logarithm in equation \eqref{eq:cap} (and in
the rest of this article) is taken with respect to base 2, i.e. $\log \equiv
\log_2$. This way, the one-shot quantum capacity corresponds to the maximal
number of qubits that can be stored and retrieved with a fidelity of at least
$1-\epsilon$ using the best possible error correcting code.

The one-shot quantum capacity tells us strictly more than the asymptotic quantum
capacity, in the sense that the latter can be obtained from the former:
\begin{align}
  \label{eq:asympcap}
  \Q(\chan) = \lim_{\epsilon \to 0} \lim_{N \to \infty} \frac{1}{N}
  \Q^\epsilon(\chan^{\otimes N}) \,.
\end{align}
The asymptotic quantum capacity is the number of qubits that can be transmitted
or stored per use of a device with asymptotically vanishing error, in the limit where it is used infinitely often
under the i.i.d. assumption. 
Therefore, it is an asymptotic \emph{rate}, while
the one-shot quantum capacity is the total number of qubits that can be
transmitted or stored in a single use of a (possibly non-tensor product) channel,
allowing some error $\epsilon \geq 0$.

\begin{table}[h]
  \begin{framed}
    \begin{center}
      \textbf{One-shot quantum capacity estimation}
    \end{center}
    \hrule
    \begin{center} \textbf{Protocol parameter} \end{center}
    \begin{itemize}[leftmargin=0.4cm]
      \item $N \in \mathbb{N}$, even: total number of qubits
    \end{itemize}
    \begin{center} \textbf{The protocol} \end{center}
    \begin{itemize}[leftmargin=0.4cm]
      \item Alice chooses $s \in \{0,1\}^{N}$ and $b \in \{\xb,\zb\}^{N}_{N/2}$
	fully at random and communicates them to Bob, where
	\begin{align*}
	  \{\xb,\zb\}^{N}_{N/2} = \left\{ b \in \{\xb,\zb\}^{N} \,\middle\vert\,
	    \parbox{2.5cm}{X, Z
	  \textnormal{ each occur $N/2$ times in $b$}} \right\} \,.
	\end{align*}
      \item For each qubit slot $i = 1, \ldots, N$ of the channel, Alice
	prepares a test qubit $i$ in the state $s_i$ with respect to basis $b_i
	\in \{\xb,\zb\}$ and sends it through the channel to Bob.
      \item For each qubit $i = 1, \ldots, N$ that Bob receives, he measures
	test qubit $i$ in the basis $b_i$ and records the outcome $s'_i \in
	\{0,1\}$.
      \item Bob determine the error rates
	\begin{align*}
	  \exest = \frac{2}{N} \sum_{i \in \ix} s_i \oplus s'_i \,, \qquad
	  \ezest = \frac{2}{N} \sum_{i \in \iz} s_i \oplus s'_i \,,
	\end{align*}
	where
	\begin{align*}
	  &\ix = \{ i \in \{1, \ldots, N\} \mid b_i = \xb \} \,, \\
	  &\iz = \{ i \in \{1, \ldots, N\} \mid b_i = \zb \} \,.
	\end{align*}
      \item Knowing the two error rates $\exest$ and $\ezest$, Bob determines a
	lower bound on the one-shot quantum capacity according to
	Theorem~\ref{thm:estimation}.
    \end{itemize}
  \end{framed}
  \caption{\textbf{The estimation protocol.}
    \label{prot:estimation}}
\end{table}

\subsection*{One-shot quantum capacity estimation}

Now that the one-shot quantum capacity is identified as the relevant figure of
merit for quantum memory and communication devices, the question is whether we
can estimate this quantity for a given device.
We answer this question in the affirmative for the case where
$\chan$ is a channel that stores or communicates (arbitrarily many) qubits.

We present a simple protocol (see Protocol~\ref{prot:estimation}) that estimates the
one-shot quantum capacity $\Q^\epsilon(\chan)$ for an $N$-to-$N$-qubit channel
$\chan$. Our protocol only requires the preparation and measurement of single qubit states in two bases.
Specifically, even though it is known that the optimal encoder for a given channel $\Lambda$ may require
the creation of a highly entangled state, no entanglement is required to execute our test.
For simplicity, we assume here that $N$ is an even number (for more
general cases, see Section~\ref{app:final}). The protocol does not make any assumption
on whether the qubits are processed sequentially, as in communication devices,
or in parallel, as in storage devices (potentially with correlated errors in
both cases). The data collection of the protocol is
very simple. Alice and Bob agree on two qubit bases $\xb$ and $\zb$. These two
bases should be chosen to be ``incompatible'', in the sense that the
\emph{preparation quality} $q$, which is defined as
\begin{align}
  \label{eq:q}
  q = -\log \max_{i,j=0,1} \abs{ {\langle i_\xb | j_\zb \rangle } }^2 \,,
\end{align}
is as high as possible, where $\ket{i_\xb}$ and $\ket{j_\zb}$ are eigenstates of
$X$ and $Z$, respectively. In the ideal case, where the two bases $\xb$ and
$\zb$ are mutually unbiased bases (MUBs), such as the Pauli-$X$ and $Z$ basis,
it holds that $q=1$. Our protocol can be seen as exploiting the idea that the
ability to transmit information in two complementary bases relates to a
channel's ability to convey (quantum) information~\cite{RB09,CW05}, which
we show holds even with correlated noise. We remark that Pauli-$X$ and $Z$ basis
have also been used to estimate the process fidelity of a quantum operation~\cite{hofman,hofmanFollowup}
in the i.i.d. case, which however we are precisely trying to avoid here.

The bound for the capacity estimate is a function of the number of qubits
$N$, the preparation quality $q$, the maximally allowed decoding error
probability $\epsilon$ of $\Q^\epsilon(\chan)$, the two measured error rates
$\exest$ and $\ezest$, and some probability $p$ that quantifies the typicality
of the protocol run (we will discuss this parameter in the Discussion section).
More precisely, the bound is given as follows.

\begin{thm}
  \label{thm:estimation}
  Let $N \in \mathbb{N}_+$ be an even number, let $\exest$ and $\ezest$ be error
  rates determined in a run of Protocol~\ref{prot:estimation} where the used bases $X$
  and $Z$ had a preparation quality of $q$ (see equation \eqref{eq:q} above).
  Then, for every $\epsilon > 0$ and for every $p \in [0,1)$, it holds that
  \begin{itemize}
    \item either, the probability that at least one error rate exceeds
      $\exest$ or $\ezest$, respectively, was higher than $p$,
    \item or the one-shot quantum capacity of the $N$-qubit channel
      $\chan$ is bounded by
  \end{itemize}
  \pagebreak
  \begin{widetext}
      \begin{align}
	\label{eq:bound}
	\Q^\epsilon(\chan) \geq \sup_{\eta \in \left(0,\sqrt{\epsilon/2}\right)}
	\Bigg[ N \Big( q - h\left(\exest + \mu\right) - h\left(\ezest + \mu\right)
	\Big) - 2 \log\left( \kappa \right) - 4 \log \left( \frac{1}{\eta} \right)
	- 2 \Bigg] \,,
      \end{align}
  \end{widetext}
  where $h$ is the binary entropy function
  \begin{align}
    h(x) := - x \log (x) - (1-x) \log (1-x)
  \end{align}
  and $\mu$ and $\kappa$ are given by
  \begin{align}
    \label{eq:kappa}
    \mu = \sqrt{ \frac{N+2}{N^2} \ln \left(
      \frac{ 3+\frac{5}{\sqrt{1-p}} }{ \sqrt{\epsilon / 2} -
      \eta} \right) } \,, \quad
    \kappa = 2 \left(\frac{3+\frac{5}{\sqrt{1-p}}}
      {\sqrt{\epsilon/2} - \eta}\right)^2 \,.
  \end{align}
\end{thm}

In the asymptotic limit where $N \to \infty$, the bound on the right hand side
of inequality \eqref{eq:bound} converges to $N(q - h(\exest) - h(\ezest))$. All
the other terms can be seen as correction terms that account for finite-size
effects. We will discuss this in more detail in the Discussion section below.

\subsection*{One-shot capacity verification}

Protocol~\ref{prot:estimation} above estimates how much quantum information can be
stored in a quantum memory device. This is of great use when the task is to
figure out whether a device is potentially useful as a quantum memory device.
When eventually, an error-correcting code is implemented, the corrected memory
might be used without further testing.

In some cases, however, one wants to implement the memory with a means to verify
its quality while using it. For example, one may suspect the quality of the
memory to diminish (say, due to damage or overuse). In that case, the capacity
estimation that was made before the implementation of the error-correcting code
may no longer be valid. A method to verify that the quality of the memory is
good enough for the implemented code may be required whenever it is used.
Protocol~\ref{prot:verification} shows such a verification protocol.

\begin{table}
  \begin{framed}
    \begin{center}
      \textbf{One-shot quantum capacity verification}
    \end{center}
    \hrule
    \begin{center} \textbf{Protocol parameters} \end{center}
    \begin{itemize}[leftmargin=0.4cm]
      \item $N \in \mathbb{N}$: number of data qubits
      \item $\etolx, \etolz \in [0,1]$: tolerated error rate in $X$, $Z$
    \end{itemize}
    \begin{center} \textbf{The protocol} \end{center}
    \begin{itemize}[leftmargin=0.4cm]
      \item Alice chooses $s \in \{0,1\}^{3N}$ and $b \in
	\{\xb,\zb,\nb\}^{3N}_N$ fully at random and communicates them to Bob,
	where
	\begin{align*}
	  \{\xb,\zb,\nb\}^{3N}_N = \left\{ b \in \{\xb,\zb,\nb\}^{3N} \,
	    \middle\vert \, \parbox{1.5cm}{\raggedright $X$, $Z$, $D$
	      \textnormal{ occur $N$ times in $b$}} \right\}.
	\end{align*}
      \item For each qubit slot $i = 1, \ldots, {3N}$ of the channel, if $b_i
	\in \{\xb,\zb\}$, Alice prepares a test qubit $i$ in the state $s_i$
	with respect to basis $b_i \in \{\xb,\zb\}$ and sends it through the
	channel to Bob. If $b_i = \nb$, Alice uses the slot for a data qubit.
      \item For each qubit $i = 1, \ldots, {3N}$ that Bob receives, if $b_i \in
	\{\xb,\zb\}$, Bob measures test qubit $i$ in the basis $b_i$ and records
	the outcome $s'_i \in \{0,1\}$. If $b_i = D$, Bob leaves the data qubit
	untouched.
      \item They determine the error rates
	\begin{align*}
	  \ex = \frac{1}{N} \sum_{i \in \ix} s_i \oplus s'_i \,, \qquad
	  \ez = \frac{1}{N} \sum_{i \in \iz} s_i \oplus s'_i \,,
	\end{align*}
	where
	\begin{align*}
	  &\ix = \{ i \in \{1, \ldots, {3N}\} \mid b_i = \xb \} \,, \\
	  &\iz = \{ i \in \{1, \ldots, {3N}\} \mid b_i = \zb \} \,.
	\end{align*}
	If $\ex \leq \etolx$ and $\ez \leq \etolz$, they continue with the
	conclusion below. Otherwise, they abort the protocol.
      \item They conclude that the one-shot quantum capacity of the channel
	$\chan$ on the $N$ data qubits is bounded as in Theorem~\ref{thm:verification}.
    \end{itemize}
  \end{framed}
  \caption{\textbf{The verification protocol.}
    \label{prot:verification}}
\end{table}

The protocol assumes that Alice holds $N$ data qubits that she wants to send to
Bob in a way that allows her to verify the quality of the transmission. To this
end, she uses a channel for $3N$ qubits and places her $N$ data qubits in random
slots of this channel. The other $2N$ slots are used for test qubits, half of
which are prepared and measured in the $X$-basis and half of which are prepared
and measured in the $Z$-basis (just as in the estimation protocol), while Alice
and Bob leave the data qubits untouched. The error rates on the test bits allows
to infer a bound on the capacity of the channel on the data qubits. This is
shown in Figure~\ref{fig:verification}.

\begin{figure}[h!]
  \centering
  \subcaptionbox{inference structure of the estimation protocol}
  {
    \begin{tikzpicture}[xscale=1,yscale=0.7]
      \node (x) {$X$-qubits};
      \draw ($(x) + (-1.5,-0.4)$) rectangle ($(x) + (1.5 ,0.4)$);
      \node at ($(x) + (4,0)$) (z) {$Z$-qubits};
      \draw ($(z) + (-1.5,-0.4)$) rectangle ($(z) + (1.5 ,0.4)$);
      \draw[decorate,decoration={brace,amplitude=5pt,mirror,raise=0.05cm}]
	($(x) + (1.5,0.7)$) -- ($(x) + (-1.5,0.7)$) node[midway,yshift=0.6cm]
	{$\exest$};
      \draw[decorate,decoration={brace,amplitude=5pt,mirror,raise=0.05cm}]
	($(z) + (1.5,0.7)$) -- ($(z) + (-1.5,0.7)$) node[midway,yshift=0.6cm]
	{$\ezest$};
      \draw[decorate,decoration={brace,amplitude=5pt,mirror,raise=0.05cm}]
	($(z) + (1.5,2)$) -- ($(x) + (-1.5,2)$) node[midway,yshift=0.6cm]
	{$\Q^\epsilon(\chan)$};
      \draw[->] ($(z) + (0,1.8)$) to[bend right=50] ($(z) + (-1.3,2.9)$);
      \draw[->] ($(x) + (0,1.8)$) to[bend left=50] ($(x) + (1.3,2.9)$);

      \draw[decorate,decoration={brace,amplitude=5pt,raise=0.05cm}]
	($(x) + (1.5,-0.7)$) -- ($(x) + (-1.5,-0.7)$) node[midway,yshift=-0.5cm]
	{$N/2$};
      \draw[decorate,decoration={brace,amplitude=5pt,raise=0.05cm}]
	($(z) + (1.5,-0.7)$) -- ($(z) + (-1.5,-0.7)$) node[midway,yshift=-0.5cm]
	{$N/2$};
      \draw[decorate,decoration={brace,amplitude=5pt,raise=0.05cm}]
	($(z) + (1.5,-1.7)$) -- ($(x) + (-1.5,-1.7)$) node[midway,yshift=-0.5cm]
	{$N$};
    \end{tikzpicture}
  }
  \subcaptionbox{inference structure of the verification protocol}
  {
    \begin{tikzpicture}[xscale=0.7,yscale=0.7]
      \node (x) {$X$-qubits};
      \node at ($(x) + (-4,0)$) (d) {data qubits};
      \node at ($(d) + (-4,0)$) (z) {$Z$-qubits};
      \draw ($(x) + (-1.5,-0.4)$) rectangle ($(x) + (1.5 ,0.4)$);
      \draw ($(d) + (-1.5,-0.4)$) rectangle ($(d) + (1.5 ,0.4)$);
      \draw ($(z) + (-1.5,-0.4)$) rectangle ($(z) + (1.5 ,0.4)$);
      \draw[decorate,decoration={brace,amplitude=5pt,mirror,raise=0.05cm}]
	($(d) + (1.5,0.7)$) -- ($(d) + (-1.5,0.7)$) node[midway,yshift=0.6cm]
	{$\Q^\epsilon(\chan)$};
      \draw[decorate,decoration={brace,amplitude=5pt,mirror,raise=0.05cm}]
	($(z) + (1.5,0.7)$) -- ($(z) + (-1.5,0.7)$) node[midway,yshift=0.6cm]
	{$\ez \leq \etolz$};
      \draw[decorate,decoration={brace,amplitude=5pt,mirror,raise=0.05cm}]
	($(x) + (1.5,0.7)$) -- ($(x) + (-1.5,0.7)$) node[midway,yshift=0.6cm]
	{$\ex \leq \etolx$};
      \draw[->] ($(z) + (0,2)$) to[bend left=50] ($(d) + (-0.5,2)$);
      \draw[->] ($(x) + (0,2)$) to[bend right=50] ($(d) + (0.5,2)$);

      \draw[decorate,decoration={brace,amplitude=5pt,raise=0.05cm}]
	($(x) + (1.5,-0.7)$) -- ($(x) + (-1.5,-0.7)$) node[midway,yshift=-0.5cm]
	{$N$};
      \draw[decorate,decoration={brace,amplitude=5pt,raise=0.05cm}]
	($(z) + (1.5,-0.7)$) -- ($(z) + (-1.5,-0.7)$) node[midway,yshift=-0.5cm]
	{$N$};
      \draw[decorate,decoration={brace,amplitude=5pt,raise=0.05cm}]
	($(d) + (1.5,-0.7)$) -- ($(d) + (-1.5,-0.7)$) node[midway,yshift=-0.5cm]
	{$N$};
    \end{tikzpicture}
  }
  \caption{\textbf{Comparison of the inference structures of the two protocols.}
    \textbf{(a)} In the estimation protocol, all qubits are test qubits, and the
    goal is to estimate the capacity for the channel on all qubits. \textbf{(b)}
    In the verification protocol, one third of the qubits are data qubits that
    are left untouched. The remaining $2N$ qubits are test qubits, whose error
    rates allow to bound the capacity of the channel on the $N$ data qubits.
    \label{fig:comparison}}
\end{figure}

For this protocol, we denote the measured error rate in $X$ by $\ex$ and the
measured error rate in $Z$ by $\ez$.  Bob checks whether these error rates
exceed some tolerated values $\etolx$ and $\etolz$, respectively, which has been
specified before the protocol run. If one or both error rates exceed the
tolerated value, the protocol aborts because the transmission quality is
considered too low. If both error rates are below their tolerated value, Bob
concludes that the transmission was of high quality, in the sense that the
channel on the data qubits had a high one-shot quantum capacity.  This is stated
more precisely in the following theorem.

\begin{thm}
  \label{thm:verification}
  Let $N \in \mathbb{N}_+$, let $\etolx$, $\etolz \in [0,1]$. Assume that
  Protocol~\ref{prot:verification} is run successfully without abortion, where the used
  bases $X$ and $Z$ had a preparation quality of $q$. Then, for every $\epsilon
  > 0$ and for every $p \in [0,1)$, it holds that
  \begin{itemize}
    \item either, the probability that the protocol aborts was higher than
      $p$,
    \item or the one-shot quantum capacity of the channel $\chan$ on the
      $N$ data qubits is bounded by inequality~\eqref{eq:bound}, where $\kappa$
      is as in equation \eqref{eq:kappa} and where $\mu$ is given by
  \end{itemize}
  \begin{align}
    \mu = \sqrt{ \frac{2(N+1)}{N^2} \ln \left(
      \frac{ 3+\frac{5}{\sqrt{1-\pabort}} }{ \sqrt{\epsilon / 2} - \eta}
      \right) } \,.
  \end{align}
\end{thm}

The bound for the verification protocol looks formally almost identical to the
one for the estimation protocol, but there are three differences. Firstly, the
function $\mu$ has a different dependence on $N$, which is a consequence of the
different structure of the protocol as explained in Figure~\ref{fig:comparison}.
Secondly, the error rates $\etolx$ and $\etolz$ are preset accepted error rates
instead of calculated error rates from data, and the bound holds when the measured rates are below
those preset values. Thirdly, the probability $p$ in the bound is the abort
probability of the protocol. We will say more about this probability in the
Discussion section.

\section*{Discussion}

In this section, we shall discuss our bound as a bound on the \emph{rate}
$\frac{1}{N} \Q^\epsilon(\chan)$, which quantifies the amount of quantum
information that can be sent \emph{per qubit}. This has the advantage that it
makes comparisons easier. To discuss our bound on the capacity rate, we have
plotted its value as a function of $N$ in Figure~\ref{fig:plot}. We plotted the bound
for the estimation protocol, but qualitatively, the bound for the verification
protocol behaves identically, so our discussion applies to both protocols.

\begin{figure*}[t]
  \centering
  \subcaptionbox{for different decoding error probabilities $\epsilon$}
  {
    \begin{tikzpicture}
      \begin{semilogxaxis}[
	width=7.8cm,
	xmin = 1000,
	xmax = 100000000,
	ymin = 0,
	xlabel = {$N$},
	ylabel = {bound on the rate $\frac{1}{N} \Q^\epsilon(\chan)$},
	ylabel style = {xshift=-0.3cm},
	legend style = {at = {(0.35,0.3)}, anchor=west},
	legend cell align = left,
	ytick = {0, 0.1, 0.2, 0.3},
	extra y ticks = {0, 0.427206085768},
	extra y tick labels = {0,$\approx 0.427$},
      ]
	\addplot[dashed, black] table [col sep=comma]
	  {csv/asymp.csv};

	\addplot[solid, OliveGreen] table [col sep=comma]
	  {csv/rate_err-0.05_eps-E-2.csv};
	\addplot[solid, blue] table [col sep=comma]
	  {csv/rate_err-0.05_eps-E-6.csv};
	\addplot[solid, red] table [col sep=comma]
	  {csv/rate_err-0.05_eps-E-10.csv};
	\addlegendentry{$q-h(\exest)-h(\ezest)$}
	\addlegendentry{$\epsilon = 10^{-2}$}
	\addlegendentry{$\epsilon = 10^{-6}$}
	\addlegendentry{$\epsilon = 10^{-10}$}
      \end{semilogxaxis}
    \end{tikzpicture}
  }
  \subcaptionbox{for different error rates $\exest$ and $\ezest$}
  {
    \begin{tikzpicture}
      \begin{semilogxaxis}[
	width=7.8cm,
	xmin = 1000,
	xmax = 100000000,
	ymin = 0,
	ymax = 1,
	xlabel = {$N$},
	ylabel = {bound on the rate $\frac{1}{N} \Q^\epsilon(\chan)$},
	legend style = {at = {(0.05,0.77)}, anchor=west},
	legend cell align = left,
      ]

	\addplot[solid, OliveGreen] table [col sep=comma]
	  {csv/rate_err-0.02_eps-E-6.csv};
	\addplot[dashdotted, black] table [col sep=comma]
	  {csv/rate_errx-0.08_errz-0.00_eps-E-6.csv};
	\addplot[solid, blue] table [col sep=comma]
	  {csv/rate_err-0.05_eps-E-6.csv};
	\addplot[solid, red] table [col sep=comma]
	  {csv/rate_err-0.08_eps-E-6.csv};
	\addlegendentry{$\exest = \ezest = 2\%$}
	\addlegendentry{$\exest = 8\%, \ezest = 0$}
	\addlegendentry{$\exest = \ezest = 5\%$}
	\addlegendentry{$\exest = \ezest = 8\%$}
      \end{semilogxaxis}
    \end{tikzpicture}
  }
  \caption{\textbf{Bound on the rate for the capacity estimation protocol as a
      function of the number of qubits.}
    This figure shows the bound on the one-shot quantum capacity for the
    estimation protocol expressed as a \emph{rate}, that is, the right hand side
    of inequality \eqref{eq:bound} divided by the number of qubits $N$. The
    plots show the bound as a function of $N$ with the parameters as $q = 1$,
    and $p = 1/2$.
    \textbf{(a)} We plotted the bound for fixed error rates $\exest = \ezest =
    5\%$ for a few different values of $\epsilon$ in order to visualize the
    dependence on the decoding error probability. The lower the allowed decoding
    error probability $\epsilon$ is set, the higher the number of qubits needs
    to be in order to get a positive bound on the rate (note that the $N$-axis
    is logarithmic). In the asymptotic limit $N \to \infty$, the bound converges
    to $q - h(\exest) - h(\ezest)$. If $q=1$, this coincides exactly with the
    (asymptotic) capacity for some important classes of channels, such as
    depolarizing channels. This shows that our bound is asymptotically optimal,
    and therefore, improvements are only possible in the finite-size correction
    terms.
    \textbf{(b)} To see the dependence on the error rates, we plotted our bound
    for a fixed value of $\epsilon = 10^{-6}$ for a few different values of
    $\exest$ and $\ezest$. The higher the error rate, the higher the number of
    qubits needs to be in order to achieve a positive rate. For every pair of
    error rates $\exest$ and $\ezest$, the bound is monotonically increasing in
    $N$ and converges to $q-h(\exest)-h(\ezest)$. Therefore, the bound can only
    be positive when $q-h(\exest)-h(\ezest)$ is positive, which yields an easy
    criterion for the potential usefulness of a channel with known error rates
    (although the full version of the bound with the correction terms is not
    hard to evaluate either).
    \label{fig:plot}}
\end{figure*}

\smallskip
\textit{Example: Dephasing channel.} In order to assess the strength of
our bound, it is helpful consider some example channels. A particularly
insightful example is the case where the channel $\chan$ is given by $N$
independent copies of a dephasing channel of strength $\alpha \in [0,1]$, that
is,
\begin{align}
  \label{eq:dephasing}
  \chan = \chan_D^{\otimes N} \,, \quad \chan_D(\rho): \rho \mapsto
  \left(1-\frac{\alpha}{2}\right) \rho + \frac{\alpha}{2} \sigma \rho \sigma \,,
\end{align}
where $\sigma$ denotes one of the qubit Pauli operators with respect to some
basis. Of particular interest is the case where the dephasing happens with
respect to one of the two bases $X$ or $Z$ in which Alice and Bob prepare and
measure the test qubits. Let us assume that $\sigma = \sigma_Z$. In order to see
what happens when our estimation protocol is used in this case, we could
simulate a protocol run and see what bound on the one-shot quantum capacity
would be obtained. However, the estimation protocol does essentially nothing but
determine the two error rates $\exest$ and $\ezest$. The expected values of
these rates can be readily obtained from equation \eqref{eq:dephasing}. The
error rate $\ezest$ vanishes, because dephasing in the $Z$-basis leaves the
$Z$-diagonal invariant. In the $X$-basis the bits are left invariant with
probability $1-\alpha/2$, and flipped with probability $\alpha/2$, so asymptotically $\exest =
\alpha/2$. Hence, for the dephasing channel, the estimation protocol is expected
to yield the bound in inequality \eqref{eq:bound} with $\ezest = 0$ and $\exest
= \alpha / 2$.

\smallskip
\textit{Asymptotic tightness of the bound.} As one can see in Figure~\ref{fig:plot},
the bound on the one-shot quantum capacity, expressed as a rate, converges to
$q-h(\exest)-h(\ezest)$, which in the case of the dephasing channel is given by
$q - h(\alpha/2)$. If we additionally assume that the bases $X$ and $Z$ are
mutually unbiased (as are Pauli-$X$ and $Z$), this
is equal to $1 - h(\alpha/2)$. This is precisely the (asymptotic)
quantum capacity of the dephasing channel. This means that our bound on the
one-shot quantum capacity is asymptotically tight; if our bound can be improved,
then only in the finite-size correction terms. In particular, our bound cannot
be improved by a constant factor. Since most estimates that enter the derivation
of the bound are of the same type as the estimates used in modern security
proofs of QKD \cite{TLGR12}, any possible improvements of the QKD security
bounds would also lead to an improvement of our bound on the one-shot quantum
capacity (if there is any). In this sense, our bound is essentially as tight as
the corresponding security bounds for QKD in the finite regime.

\smallskip
\textit{Measurement calibration.} Above, we have assumed that Alice and
Bob were very lucky: they set up their bases $X$ and $Z$ such that one of them
is exactly aligned with the dephasing basis, and therefore optimally exploited
the asymmetry of the channel. In general, since they do not know the channel
whose capacity they estimate, they do not know about the direction of the
asymmetry. Instead, they have to calibrate their devices by trying out several
pairs of bases until they find one with low error rates. Otherwise, the bound on
the one-shot quantum capacity that they infer is suboptimal. It is an
interesting open question how such a calibration can be optimized.

\smallskip
\textit{Example: Fully depolarizing channel.}
Another insightful example is the case where $\chan$ is given by the channel
which outputs the fully mixed state of $N$ qubits, independently of the input
state. The capacity of this channel is zero, yet with probability $2^{-N}$,
Alice and Bob measure error rates $\exest = \ezest = 0$. One may think that
these vanishing error rates lead to a highly positive bound on the capacity, but
this is not the case. As one can read in Theorem~\ref{thm:estimation} and Theorem~\ref{thm:verification},
the bound depends on a probability $p$, and the term $1-p$ corresponds precisely
to the probability of such an unlikely case. In fact, for $1-p=2^{-N}$, the
bound is never positive. This example shows that in the one-shot regime, a
meaningful capacity estimation can only be made under the assumption that the
observed data is not extremely untypical for the channel. However, this is only
a problem for very low values of $N$: thanks to the natural logarithm in $\mu$
(see equation \eqref{eq:kappa} above), the concern reduces to untypical events
with an exponentially (in $N$) small probability. For reasonable numbers of $N$,
the influence of $p$ on the bound is negligible, except for extremely low values
of $1-p$. For more information on this probability, see Section~\ref{app:final}. We
note that this issue is not only given in our context of capacity estimation,
but in all statistical tests on a finite sample, including quantum key
distribution.

\begin{figure*}[t]
  \centering
  \subcaptionbox{plot for each of the three storage times}
  {
    \begin{tikzpicture}
      \begin{semilogxaxis}[
        width=7.8cm,
        xmin = 0.000000001,
        xmax = 0.1,
        ymin = -0.01,
        xlabel = {$\log_{10}(\epsilon)$},
        xtick = {0.000000001,0.00000001,0.0000001,0.000001,0.00001,0.0001,
          0.001,0.01,0.1},
        xticklabels = {$-9$,$-8$,$-7$,$-6$,$-5$,$-4$,$-3$,$-2$,$-1$},
        ylabel = {bound on the rate $\frac{1}{N} \Q^\epsilon(\chan)$},
        ylabel style = {xshift=-0.3cm},
        legend style = {at = {(0.35,0.35)}, anchor=west},
        legend cell align = left,
            ]
        \addplot[solid, OliveGreen] table [col sep=comma]
          {csv/bound_vs_epsilon_300ns.csv};
        \addplot[solid, blue] table [col sep=comma]
          {csv/bound_vs_epsilon_600ns.csv};
        \addplot[solid, red] table [col sep=comma]
          {csv/bound_vs_epsilon_1us.csv};
        \addlegendentry{$\Delta t = \SI{300}{\nano\second}$}
        \addlegendentry{$\Delta t = \SI{600}{\nano\second}$}
        \addlegendentry{$\Delta t = \SI{1}{\micro\second}$}
      \end{semilogxaxis}
    \end{tikzpicture}
  }
  \subcaptionbox{zoomed-in plot for $\Delta t = \SI{300}{\nano\second}$}
  {
    \begin{tikzpicture}
      \begin{semilogxaxis}[
        width=7.8cm,
        xmin = 0.000000001,
        xmax = 0.1,
        xlabel = {$\log_{10}(\epsilon)$},
        xtick = {0.000000001,0.00000001,0.0000001,0.000001,0.00001,0.0001,
          0.001,0.01,0.1},
        xticklabels = {$-9$,$-8$,$-7$,$-6$,$-5$,$-4$,$-3$,$-2$,$-1$},
        ylabel = {bound on the rate $\frac{1}{N} \Q^\epsilon(\chan)$},
        ylabel style = {xshift=-0.3cm},
        legend style = {at = {(0.4,0.25)}, anchor=west},
        legend cell align = left,
            ]
        \addplot[solid, OliveGreen] table [col sep=comma]
          {csv/bound_vs_epsilon_300ns.csv};
        \addlegendentry{$\Delta t = 300 \text{ ns}$}
      \end{semilogxaxis}
    \end{tikzpicture}
  }
  \caption{\textbf{Bound on the rate for the experimental data as a function of
    $\epsilon$.}
    This figure shows the bound on the one-shot quantum capacity rate for the
    data gained in the transmon qubit. We pick $p=1/2$, and use $q=0.9$ as preparation quality to account for the
experimental imperfections.
    \textbf{(a)} The experiment was carried out three times with different
    storage times $\Delta t$, for each of which we plotted the bound resulting
    from the estimation protocol as a function of the decoding error probability
    $\epsilon$. Since the number of qubit preparations and measurements was high
    ($N = \num{1.04e6}$), the dependence on $\epsilon$ is rather small.
    \textbf{(b)} For a better visibility of the $\epsilon$-dependence, we show
    the plot for the shortest storage time separately and more zoomed-in in the
    direction of the bound.
    \label{fig:experiment-plots}}
\end{figure*}
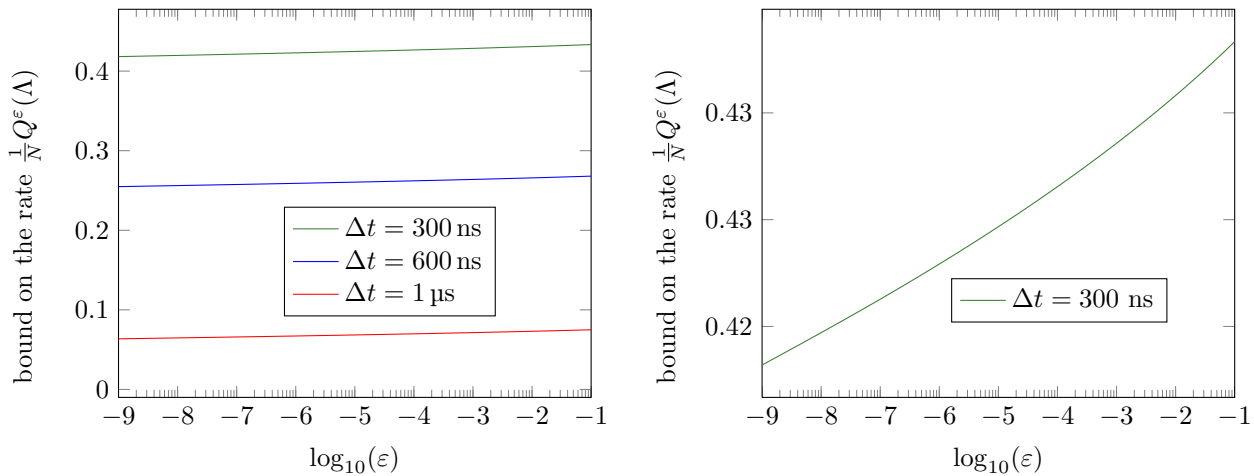

\smallskip
\textit{Experiment.} We demonstrate the use of this protocol by implementing it on a Transmon qubit.
The experiment is performed on qubit $A_{\mathrm{T}}$ previously reported on in~\cite{Riste15}.
We measure a relaxation time of $T_1 = \SI{18.5 \pm 0.6}{\micro\second}$ and a
Ramsey dephasing time of $T_2^\star = \SI{3.8 \pm 0.3}{\micro\second}$ before performing the experiment.
Readout of the qubit state is performed by probing the readout resonator with a microwave tone.
The resulting transients are amplified using a TWPA~\cite{lincolnlabssciencepaper} at the front end of the amplification chain.
This results in a readout fidelity $F_{\mathrm{RO}} = 1-({p_{01} +p_{10}})/{2}=98.0\%$,
where $p_{01}$ ($p_{10}$) is the probability of declaring state 1 (0) when the input state was $\ket{0}$ ($\ket{1}$) respectively.
The qubit state is controlled using resonant microwave pulses.

The experiment implements Protocol~\ref{prot:estimation} to estimate the capacity of the idling operation $I(\tau)$. 
We do this by generating 8000 pairs of random numbers corresponding to the bases $b\in \{X, Z \}$ and states $s\in \{0,1\}$.
These are then used to generate pulse sequences that rotate $\ket{0}$ to the required state, and wait for a time $\tau$ before measuring the qubit in the $Z$-basis and declaring a state.
If the required state was in the $X$ basis, a recovery pulse is applied that rotates the state to the $Z$ basis before it is read out.
This protocol is repeated 130 times, with a distinct randomization for each repetition, yielding a total of $N = \num{1.04e6}$ measurement outcomes in approximately one and a half hours.




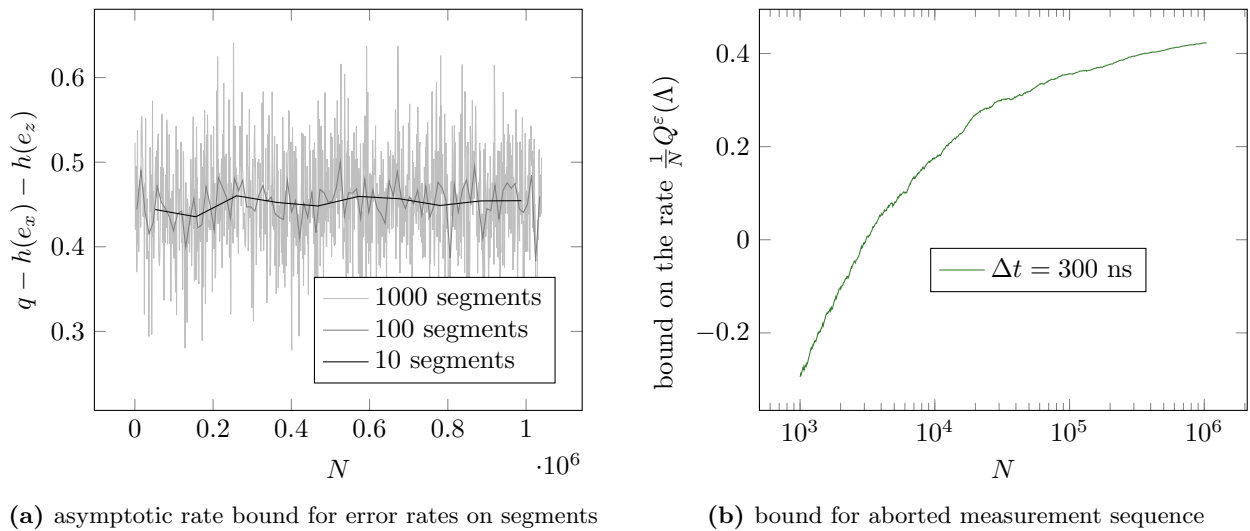
\begin{figure*}[t]
  \centering
  \subcaptionbox{asymptotic rate bound for error rates on segments}
  {
    \begin{tikzpicture}
      \begin{axis}[
        width=8cm,
        xlabel = {$N$},
        ylabel = {$q - h(e_x) - h(e_z)$},
        ylabel style = {yshift=-0.3cm},
        legend style = {at = {(0.45,0.2)}, anchor=west},
        legend cell align = left,
            ]
        \addplot[solid, lightgray] table [col sep=comma] {csv/fluc_1000.csv};
        \addplot[solid, gray] table [col sep=comma] {csv/fluc_100.csv};
        \addplot[solid, black] table [col sep=comma] {csv/fluc_10.csv};
        \addlegendentry{1000 segments}
        \addlegendentry{100 segments}
        \addlegendentry{10 segments}
      \end{axis}
    \end{tikzpicture}
    \hspace{0.3cm}
  }
  \subcaptionbox{bound for aborted measurement sequence}
  {
    \begin{tikzpicture}
      \begin{semilogxaxis}[
        width=8cm,
        xlabel = {$N$},
        ylabel = {bound on the rate $\frac{1}{N} \Q^\epsilon(\chan)$},
        ylabel style = {xshift=-0.3cm},
        legend style = {at = {(0.35,0.35)}, anchor=west},
        legend cell align = left,
            ]
        \addplot[solid, OliveGreen] table [col sep=comma]
          {csv/bound_vs_N.csv};
        \addlegendentry{$\Delta t = 300 \text{ ns}$}
      \end{semilogxaxis}
    \end{tikzpicture}
    \hspace{0.3cm}
  }
  \caption{\textbf{Error fluctuations across the measurements.}
    Here we visualize the statistical fluctuations in the measurement outcomes
    over the course of the transmon qubit experiment.
    \textbf{(a)} For the experiment with $\Delta t = \SI{300}{\nano\second}$, we
    split up the $N = \num{1.04e6}$ sequential measurement outcomes into equally
    large and chronologically ordered segments and calculate the error rates
    $e_x$ and $e_z$ on each segment. For a meaningful and comparable quantity
    for comparison, we calculate the asymptotic bound $q-h(e_x)-h(e_z)$ for each
    of segment with $q=0.9$, that is, the bound on the capacity rate that would be obtained
    if infinitely many measurements with the error rates as on the respective
    segments would be measured. As expected, the fluctuations decrease with the
    number of segments, or in other words, the larger the segments, the smaller
    the differences between them. Note that in contrast to all other plots, this
    is a linear plot.
    \textbf{(b)} For a glimpse on the cumulative effect of the fluctuations, we
    set 1000 logarithmically distributed ``break points'' and calculate the
    bound as if the experiment ended at each of those points where $q=0.9$, $\epsilon=10^{-6}$, and we pick $p=1/2$. The resulting plot
    is to be compared with the plots in Figure~\ref{fig:plot}. The fluctuations that
    make the curve deviate from a smooth curve come from the fact that the
    measured error rates are not constant throughout the experiment, indicating that the noise affecting the transmon qubit is indeed unlikely
	to correspond to an i.i.d. process.
    \label{fig:fluc}}
\end{figure*}

\smallskip
\textit{Other open questions.}
Our result assumes that the system on which the channel acts is composed of
qubits. An interesting open question is whether this restriction can be removed
and an analogous bound can be derived for channels of arbitrary dimension and
composition.

It would also be interesting to see our bound extended to continuous variable systems.
There are many tools already available~\cite{FFBL12,furrer,berta,beta:ur} that may be useful to perform such
an analysis, but it remains to be determined how exactly they can be applied to such systems.

\section*{Methods}
\label{sec:methods}

To prove the bound on the one-shot quantum capacity, we combine several results.
Firstly, as we recapitulate in more detail in Section~\ref{app:bound-on-capacity}, it
has been shown that the one-shot quantum capacity is bounded by the one-shot
capacity of \emph{entanglement transmission} $\Qent^\epsilon(\chan)$
\cite{BKN00}. More precisely, it holds that for every channel $\chan$ and for
every $\epsilon > 0$ \cite{BD10},
\begin{align}
  \label{eq:cap-similar}
  \Q^\epsilon(\chan) \geq \Qent^{\epsilon/2}(\chan) - 1 \,.
\end{align}
The one-shot capacity of entanglement transmission, in turn, has been proved to
be bounded by the smooth min-entropy $\Hmin^{\epsilon}(A|E)_\rho$, which is
defined by
\begin{align}
  \Hmin^\epsilon(A|B)_\rho := \max_{\rho' \in B^\epsilon(\rho)}
  \Hmin(A|B)_{\rho'} \,,
\end{align}
where
\begin{align}
	\Hmin(A|B)_\rho := \max_{\sigma_B} \sup \{
	\lambda \in \mathbb{R} \mid \rho_{AB} \leq 2^{-\lambda} I_A \otimes
	\sigma_B \} \,.
\end{align}
It has been shown that \cite{BD10,MW14,TBR16}
\begin{align}
  \label{eq:minentbound}
  \Q^\epsilon_{\text{ent}}(\chan) \geq \sup_{\eta \in (0,\sqrt{\epsilon})}
  \left( \Hmin^{\sqrt{\epsilon} - \eta}(A|E)_\rho - 4 \log \frac{1}{\eta} - 1
  \right) \,,
\end{align}
Here, the smooth min-entropy is evaluated for the state
\begin{align}
  \rho_{AE} = (\idchan{A} \otimes \chan^c_{A' \rightarrow E}) (\Phi_{AA'}) \,,
\end{align}
where $\Phi_{AA'}$ is a maximally entangled state over the input system $A'$ and
a copy $A$ of it, and where $\chan^c_{A' \rightarrow E}$ is the
\emph{complementary channel} of the channel $\chan_{A' \rightarrow B}$. The
system $E$ is the environment of the channel (see \cite{Wil13,TBR16} and
Section~\ref{app:bound-on-capacity} for more details). Taking together the results
\eqref{eq:cap-similar} and \eqref{eq:minentbound}, we get that for all $\epsilon
> 0$,
\begin{align}
  \Q^\epsilon(\chan) \geq \sup_{\eta \in (0, \sqrt{\epsilon / 2})} \left(
    \Hmin^{\sqrt{\epsilon / 2} - \eta}(A|E)_\rho - 4 \log \frac{1}{\eta} - 2
    \right) \,.
\end{align}
Therefore, the min-entropy bounds the one-shot quantum capacity.

Estimating the min-entropy has been a subject of intense research in quantum key
distribution (QKD). However, min-entropy estimation protocols in QKD cannot be
directly applied here, because they estimate the min-entropy
$\Hmin^\epsilon(X|E)$ for \emph{classical} information $X$, while in the bound
\eqref{eq:minentbound}, the system $A$ holds \emph{quantum information}. We
bridge this gap: as our main technical contribution, we show in
Section~\ref{app:boundonhmin} that for a system $A$ that is composed of qubits, it
holds that for every $\epsilon > 0$ and every $\epsilon', \epsilon'' \geq 0$,
\begin{align}
  \label{eq:main-contribution}
  &H_{\text{min}}^{3\epsilon + \epsilon' + 4\epsilon''}(A|E)_\rho \nonumber \\
  & \quad \geq Nq
    - \left( H_{\text{max}}^{\epsilon''}(X|B)_\rho
    + H_{\text{max}}^{\epsilon'}(Z|B)_\rho \right)
    - 2\log\frac{2}{\epsilon^2} \,.
\end{align}
Inequality \eqref{eq:main-contribution} reduces estimating the min-entropy of
quantum information $A$ to estimating the max-entropy of measurement outcomes
$X$ and $Z$ on the system $A$.

We prove inequality \eqref{eq:main-contribution} using three main ingredients.
Firstly, we use an uncertainty relation for the smooth min- and max-entropies
\cite{TR11}. Secondly, we use a duality relation for the smooth min- and
max-entropies \cite{KRS09,TCR10}. These two ingredients were also used in modern
security proofs of quantum key distribution \cite{TLGR12}. We combine these two
tools with a third tool, namely a chain rule theorem for the smooth max-entropy
\cite{VDTR12} to arrive at the bound in inequality \eqref{eq:main-contribution}.

Given inequalities \eqref{eq:minentbound} and \eqref{eq:main-contribution}, all
we are left to do is to devise a protocol that estimates the max-entropies of
$X$ and $Z$ given Bob's quantum information $B$. Here we can make use of
protocols in quantum key distribution that estimate exactly such a quantity. We
show in Section~\ref{app:final} how two such protocols (one for the max-entropy of $X$
and one for the max-entropy of $Z$) can be combined into one protocol, which
estimates both quantities simultaneously. The resulting protocol, which we
presented in two versions, is given by Protocol~\ref{prot:estimation} and
Protocol~\ref{prot:verification} in the Results section. Our bound on the one-shot
quantum capacity of the channel, inequality \eqref{eq:bound}, is obtained by
combining inequalities \eqref{eq:cap-similar} and \eqref{eq:minentbound} with
these max-entropy estimation techniques.

\section*{Acknowledgements}

We thank the Leo DiCarlo group (incl. MAR) at QuTech for providing test data for a superconducting qubit.
We also thank
Mario Berta, Thinh Le Phuc and Jeremy Ribeiro for insightful
discussions, and Julia Cramer for helpful comments on an earlier version of this manuscript. CP, AM, MT and SW were supported by MOE Tier 3A grant ''Randomness
from quantum processes'', NRF CRP ''Space-based QKD''.
SW was also supported by
STW, Netherlands, an NWO VIDI Grant, and an ERC Starting Grant.
MAR was supported by an ERC Synergy grant.

\bibliographystyle{plain}

\bibliography{capTomo}

\onecolumn 

\appendix

\section{Mathematical preliminaries and conventions}
\label{app:preliminaries}

In appendix~\ref{app:preliminaries}, we recapitulate some basic definitions. The material presented
should not be seen as an introduction to the subject. We only state the
definitions here to avoid ambiguity and to clarify our notation. The following
definition clarifies our notation for operator spaces.

\begin{defn}[Operator spaces]
  For a finite-dimensional complex inner product space $\ca{H}$, we define the
  following sets of operators on $\ca{H}$:
  \begin{itemize}
    \item $\End(\ca{H}) := \left\{ L: \ca{H} \rightarrow \ca{H} \mid L
	\textnormal{ linear} \right\}$, \hfill (endomorphisms on $\ca{H}$)
    \item $\Herm(\ca{H}) := \left\{ L \in \End(\ca{H}) \mid L^\dagger = L
      \right\}$, \hfill (Hermitian operators on $\ca{H}$)
    \item $\ca{S}(\ca{H}) := \{ \rho\in\Herm(\ca{H})\mid\rho \geq 0, \tr(\rho)=
      1 \}$, \hfill (states / density operators on $\ca{H}$)
    \item $\ca{S}^\leq(\ca{H}) := \{ \rho \in \Herm(\ca{H}) \mid \rho \geq 0, \tr
      (\rho) \leq 1 \}$. \hfill (subnormalized states on $\ca{H}$)
  \end{itemize}
\end{defn}

Next, we make some general conventions.

\begin{conv}
  Throughout this document, we make use of the following conventions:
  \begin{itemize}
    \item $\log$ denotes the binary logarithm (base 2) and $\ln$ denotes the
      natural logarithm (base $e$).
    \item Quantum systems are assumed to be finite-dimensional, and the symbol
      $\ca{H}$ always denotes a finite-dimensional complex inner product space.
    \item A single subscript of $\ca{H}$ refers to the system associated with
      the space (for example, $\ca{H}_A$ is the space of system~$A$).
    \item We use multiple subscripts of $\ca{H}$ to refer to a space of a joint
      system (for example, $\ca{H}_{AB} = \ca{H}_A \otimes \ca{H}_B$).
   \item For multipartite states, such as $\rho_{ABE} \in \ca{S}(\ca{H}_{ABE})$,
     we denote its reduced states by according changes of the subscript, e.g.
     $\rho_{AE} := \tr_B(\rho_{ABE})$, $\rho_A := \tr_B(\tr_E(\rho_{ABE}))$.
  \end{itemize}
\end{conv}

The formal definition of a quantum channel goes as follows.

\begin{defn}[Quantum channel]
  Let $A$ and $B$ be quantum systems.
  \begin{itemize}
    \item The \dt{identity channel} on $A$, denoted by $\idchan{A}$, is the
      linear map
      \begin{align}
	\begin{array}{cccc}
	  \idchan{A}: & \End(\ca{H}_A) & \rightarrow & \End(\ca{H}_A) \\
	         & \rho_A         & \mapsto     & \rho_A \,.
	 \end{array}
       \end{align}
    \item A \dt{quantum channel} from $A$ to $B$ is a linear map
      \begin{align}
	\chan: \End(\ca{H}_A) \rightarrow \End(\ca{H}_B)
      \end{align}
      which is trace-preserving, i.e.\
      \begin{align}
	\tr(\chan(\rho_A)) = \tr(\rho_A) \quad \Forall \rho_A \in
	\End(\ca{H}_A) \,,
      \end{align}
      and which is completely positive. That is, for any quantum system $E$ of
      any dimension $d_E \in \mathbb{N}_+$, the map $\chan \otimes \idchan{E}$
      is a positive map,
      \begin{align}
	(\chan \otimes \idop{E})(\rho_{AE}) \geq 0 \quad \Forall
	\rho_{AE} \in \ca{H}_{AE} \,.
      \end{align}
      Such a map is called a \dt{trace-preserving completely positive map}
      (abbreviated as \dt{TPCPM}).\footnote{
	According to this definition, the terms ``channel'' and ``TPCPM'' are
	equivalent. In practice, the term channel is preferred when speaking of
	the evolution of a system in a physical sense, while the term TPCPM
	refers to the map as a mathematical object. However, this distinction is
	often not very strict.}
  \end{itemize}
\end{defn}

In order to define the the one-shot quantum capacity of a channel and the smooth
entropies of quantum channels below, we need to make use of some distance
measures.

\begin{defn}[Distance measures]
  On the above operator spaces, we define the following distance measures:
  \begin{itemize}
    \item The \dt{trace norm} on $\End(\ca{H})$ is defined as $\norm{L}_1 =
      \tr\left(\sqrt{L^\dagger L}\right)$.
    \item The \dt{trace distance} on $\End(\ca{H})$ is defined as $D(\rho,
      \sigma) := \frac{1}{2} \norm{\rho-\sigma}_1$.
    \item The \dt{generalized fidelity}~\cite{TCR10} on $\ca{S}^\leq(\ca{H})$ is defined as
      $F(\rho, \sigma) := \left\Vert \sqrt{\rho} \sqrt{\sigma} \right\Vert_1 +
      \sqrt{(1-\tr\rho)(1-\tr\sigma)}$.
    \item The \dt{fidelity} is given by the restriction of the generalized
      fidelity to $\ca{S}(\ca{H})$, resulting in $F(\rho, \sigma) =
      \norm{\sqrt{\rho} \sqrt{\sigma}}$.
    \item The \dt{purified distance} on $\ca{S}^\leq(\ca{H})$ is defined as
      $P(\rho, \sigma) := \sqrt{1-F(\rho,\sigma)^2}$.
    \item The \dt{$\mathbf{\epsilon}$-ball} around a subnormalized state $\rho
      \in \ca{S}^\leq(\ca{H})$ is given by $B^\epsilon(\rho) := \{ \rho' \in
      \ca{S}^\leq(\ca{H}) \mid P(\rho, \rho') \leq \epsilon \}$.
  \end{itemize}
\end{defn}

Next, we define two kinds of one-shot capacities for quantum channels. There are
(at least) two meaningful definitions of a capacity. In the i.i.d. scenario,
these two capacities happen to coincide, and they are just referred to as
\emph{the} quantum capacity of a quantum channel. In the one-shot case, however,
the two capacities are distinct, so it is not a priori clear which one should be
chosen as \emph{the} one-shot quantum capacity $\Q^\epsilon(\chan)$. Here, we
follow Buscemi and Datta \cite{BD10}, identifying the one-shot capacity of
minimum output fidelity as the one-shot quantum capacity. More precisely, we
define the following.

\begin{defn}[One-shot capacities]
  \label{def:capacities}
  Let $\chan: \End(\ca{H}_A) \rightarrow \End(\ca{H}_B)$ be a quantum channel,
  let $\epsilon \geq 0$.
  \begin{itemize}
    \item The \dt{one-shot capacity of minimum output fidelity}
      $\Q^\epsilon(\chan)$ of the channel with respect to $\epsilon$, which we
      also call the \dt{one-shot quantum capacity} of the channel, is defined as
      \begin{align}
	\label{eq:capApp}
	\Q^\epsilon(\chan) := \max \{ \log m \mid F_\text{min}(\chan, m) \geq 1 -
	\epsilon \} \,,
      \end{align}
      where
      \begin{align}
	\label{eq:fminApp}
	F_\text{min}(\chan, m) :=
	\max_{\ca{H}_A' \subseteq \ca{H}_A \atop \dim\left(\ca{H}_A'\right) = m}
	\max_{\ca{D}} \min_{\ket{\phi} \in \ca{H}_A'} \bra{\phi} (\ca{D} \circ
	\chan)(\phi) \ket{\phi} \,.
      \end{align}
      The inner maximization ranges over all channels $\ca{D}: \End(\ca{H}_B)
      \rightarrow \End(\ca{H}_A)$ (decoding channels).
    \item The \dt{one-shot capacity of entanglement transmission}
      $\Qent^\epsilon(\chan)$ of the channel is defined as
      \begin{align}
	\Qent^\epsilon(\chan) :=
	\max \{ \log m \mid F_\text{ent}(\chan, m) \geq 1 - \epsilon \} \,,
      \end{align}
      where
      \begin{align}
	F_\text{ent}(\chan, m) \nonumber := \max_{\ca{H}_M \subseteq \ca{H}_A
	  \atop \dim(\ca{H}_M) = m} \max_{\ca{D}} \bra{\Phi_{MM'}} (\ca{D} \circ
	  \chan)(\Phi_{MM'}) \ket{\Phi_{MM'}}
      \end{align}
      The maximization over $\ca{D}$ is as above, and the state
      \begin{align}
	\ket{\Phi_{MM'}} = \frac{1}{\sqrt{\dim(\ca{H}_M)}}
	\sum_{i=1}^{\dim(\ca{H}_M)} \left( \ket{i}_M \otimes \ket{i}_{M'} \right)
      \end{align}
      is a maximally entangled state on the subsystem $\ca{H}_M$ and a copy
      $\ca{H}_{M'}$ of it.
  \end{itemize}
\end{defn}

Although these two capacities are distinct, we will see below that they are
comparable in the sense that they bound each other (see inequality
\eqref{eq:capacities-comparable} below). It is important to note that in
Definition~\ref{def:capacities}, we follow the definitions in \cite{BD10} and \cite{TBR16}
by using the fidelity as the figure of merit.  Other sources, such as
\cite{MW14}, define the one-shot capacities with the trace distance as the
figure of merit. This will have consequences in Section~\ref{app:bound-on-capacity},
when a bound on $\Qent^\epsilon(\chan)$ is converted from one definition to
another (see inequality \eqref{eq:qentbound} below).

Now we define the (smooth) min- and max-entropy of a bipartite quantum state.

\begin{defn}[Min- and max-entropy]
  \label{def:min-max}
  Let $\rho_{AB} \in \ca{S}^\leq(\ca{H}_{AB})$ be a subnormalized bipartite
  state.
  \begin{itemize}
    \item The \dt{min-entropy} of $A$ conditioned on $B$ for the state
      $\rho_{AB}$ is defined~\cite{Ren05} as
      \begin{align}
	\Hmin(A|B)_\rho := \max_{\sigma_B \in \ca{S}^{\leq}(\ca{H}_B)} \sup \{
	\lambda \in \mathbb{R} \mid \rho_{AB} \leq 2^{-\lambda} I_A \otimes
	\sigma_B \} \,.
      \end{align}
    \item The \dt{max-entropy} of $A$ given $B$ for the state $\rho_{AB}$ is
      defined~\cite{KRS09} as
      \begin{align}
	\Hmax(A|B)_\rho := \max_{\sigma_B \in \ca{S}^{\leq}(\ca{H}_B)}
	\log\left\Vert \sqrt{\rho_{AB}} \sqrt{I_A \otimes \sigma_B}
	\right\Vert^2_1 \,.
      \end{align}
  \end{itemize}
\end{defn}

\begin{defn}[Smooth min- and max-entropy]
  \label{def:smooth-entropies}
  Let $\rho_{AB} \in \ca{S}^\leq(\ca{H}_{AB})$ be a bipartite state and let
  $\epsilon \geq 0$.
  \begin{itemize}
    \item The \dt{$\epsilon$-smooth max-entropy} of $A$ conditioned on $B$ is
      defined as
      \begin{align}
        \Hmax^\epsilon(A|B)_\rho
        := \min_{\rho' \in B^\epsilon(\rho)} \Hmax(A|B)_{\rho'} \,,
      \end{align}
    \item The \dt{$\epsilon$-smooth min-entropy} of $A$ conditioned on $B$ is
      defined as
      \begin{align}
        \Hmin^\epsilon(A|B)_\rho := \max_{\rho' \in B^\epsilon(\rho)}
        \Hmin(A|B)_{\rho'} \,.
      \end{align}
  \end{itemize}
\end{defn}

For states that are defined on more systems than labeled in the entropy, the
entropy is evaluated for the according reduced state. For example, given a state
$\rho_{ABE} \in \ca{H}_{ABE}$, the smooth min-entropy $\Hmin^\epsilon(A|E)_\rho$
is evaluated for $\rho_{AE} = \tr_B(\rho_{ABE})$.

To avoid confusion with other sources that define the smooth min- and
max-entropies, it is important to note two things.
\begin{itemize}
  \item Firstly, the max-entropy, as we defined it in Definition~\ref{def:min-max},
    coincides with the R{\'e}nyi entropy of order~1/2, whereas in some older
    sources, it was defined as the R{\'e}nyi entropy of order 0 \cite{Ren05}.
  \item Secondly, the smooth entropies, as we defined them in
    Definition~\ref{def:smooth-entropies}, measure the distance in the purified distance,
    whereas in some older sources, it was defined with respect to the trace
    distance \cite{Ren05}.
\end{itemize}
There are several reasons for making the definitions as we use them here. One
important reason is that this way, the smooth min- and max-entropies satisfy a
duality relation \cite{TCR10} that we will exploit in Section~\ref{app:boundonhmin}
(see Lemma~\ref{lem:duality}).

\section{Background: Proof of the min-entropy bound on the one-shot quantum capacity}
\label{app:bound-on-capacity}

In appendix~\ref{app:bound-on-capacity}, we explain the details of the min-entropy bound on the
one-shot quantum capacity (inequality~\eqref{eq:minentbound} of the main article)
and show how it is derived. This is not a new result, but an application of
results that are well-established in quantum information science, which we provide here for convenience of the reader.

As mentioned in the main article, as the first step in the derivation of
inequality \eqref{eq:minentbound}, we note that the one-shot quantum capacity
$\Q^\epsilon(\chan)$ of a quantum channel $\chan$ can be lower-bounded by the
one-shot capacity of entanglement transmission $\Qent^\epsilon(\chan)$. More
precisely, Barnum, Knill and Nielsen \cite{BKN00} have shown that (here we use
the form presented in \cite{BD10})
\begin{align}
  \label{eq:capacities-comparable}
  \Forall \epsilon > 0: \Qent^\epsilon(\chan) - 1 \leq \Q^{2\epsilon}(\chan)
  \leq \Qent^{4\epsilon}(\chan) \,.
\end{align}
In particular,
\begin{align}
  \label{eq:capbound}
  \Forall \epsilon > 0: \quad \Q^\epsilon(\chan) \geq \Qent^{\epsilon /
  2}(\chan) - 1 \,.
\end{align}

In the next step, we will bound $\Qent^\epsilon(\chan)$. Before we do that, it
is helpful to extend our picture with the Stinespring dilation of the channel
and a purification of the input state, as shown in
Figures~\ref{fig:dilation} and~\ref{fig:fully-purified}. Readers who are already familiar with
these concepts may skip this part and continue reading below
Figure~\ref{fig:fully-purified}.

Recapitulate the situation that we consider: we are given a quantum channel
$\chan$ that takes a quantum system on Alice's side as its input and outputs
another quantum system on Bob's side. It is helpful to give these input and
output systems their own labels. We denote the input system on Alice's side by
$A'$ and the output system on Bob's side by $B$ (the reason for choosing $A'$
instead of $A$ will become clear below).
The situation is depicted in Figure~\ref{fig:chan}.

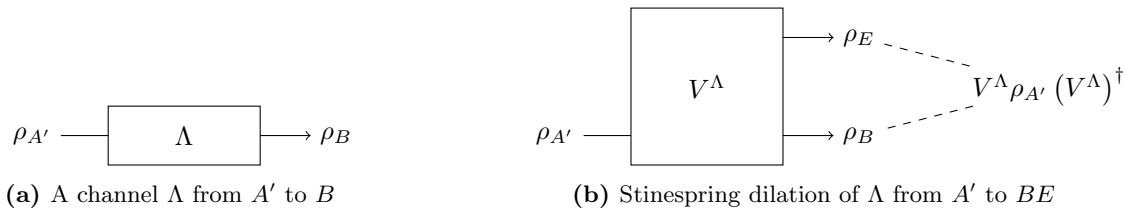
\begin{figure}[h!]
  \centering
  \subcaptionbox{
    A channel $\chan$ from $A'$ to $B$ \label{fig:chan}
  }[0.48\textwidth]
  {
    \begin{tikzpicture}[
      yscale=1.3,
    ]
      \node (a') {$\rho_{A'}$};
      \node at ($(a') + (4,0)$) (b) {$\rho_B$};
      \draw[->] (a') -- (b);

      \draw[fill=white] ($(a') + (1,-0.3)$) rectangle ($(a') + (3,0.3)$);
      \node at ($(a') + (2,0)$) {$\chan$};
    \end{tikzpicture}
  }
  \subcaptionbox{
    Stinespring dilation of $\chan$ from $A'$ to $BE$ \label{fig:stinespring}
  }[0.48\textwidth]
  {
    \begin{tikzpicture}[yscale=1.3]
      \node (a') {$\rho_{A'}$};

      \node at ($(a') + (4,0)$) (b) {$\rho_B$};
      \node at ($(b) + (0,1)$) (e) {$\rho_E$};
      \draw[->] (a') -- (b);
      \draw[->] ($(e) + (-2,0)$) -- (e);

      \draw[fill=white] ($(a') + (1,-0.3)$) rectangle ($(e) + (-1,0.3)$);
      \node at ($(a') + (2,0.5)$) {$V^\chan$};

      \node[inner sep=0] at ($(b) + (2.5,0.5)$) (v) {$V^\chan \rho_{A'} \left(
	V^\chan \right)^\dagger$};
      \draw[dashed] (e) -- (v) -- (b);
    \end{tikzpicture}
  }
  \caption{\textbf{An arbitrary quantum channel $\chan$ and its Stinespring
    dilation $V^\chan$.} For every channel $\chan$ from $A'$ to $B$, there is
    an isometry $V^\chan$ such that $\chan(\rho_{A'}) = \tr(V^\chan \rho_{A'}
    (V^\chan)^\dagger)$. Isometries have the important property that they map
    pure input states to pure output states. This will be important below, where
    we purify the input state $\rho_{A'}$ (see Figure~\ref{fig:fully-purified}).
    \label{fig:dilation}}
\end{figure}

For our purposes, it is useful to extend this picture. Mathematically speaking,
a quantum channel is a trace-preserving completely positive map that maps
density operators $\rho_{A'}$ to density operators $\rho_B$,
\begin{align}
  \chan: \ca{S}(\ca{H}_{A'}) \rightarrow \ca{S}(\ca{H}_B) \,.
\end{align}
The \emph{Stinespring dilation theorem} \cite{Sti55} states that for every such
completely positive map $\chan$, there is a system~$E$ of dimension $d_E \leq
d_{A'}^2$ and a linear isometry
\begin{align}
  V^\chan: \ca{H}_{A'} \rightarrow \ca{H}_B \otimes \ca{H}_E
\end{align}
such that
\begin{align}
  \chan(\rho_{A'}) = \tr_E \left( V^\chan \rho_{A'} \left( V^\chan \right)^\dagger
  \right) \,.
\end{align}
The extended picture is shown in Figure~\ref{fig:stinespring}. The map $V^\chan$ is an
\emph{isometric extension} or \emph{Stinespring dilation} of the channel
$\chan$, and is a standard tool in quantum information theory \cite{Wil13}.

We extend this picture further using another standard tool in quantum
information. The input state $\rho_{A'}$ of the channel may not be a pure state
but a mixed state. This is inconvenient, as many useful mathematical statements
require the involved states to be pure. However, we can work around this by
considering a \emph{purification} of $\rho_{A'}$, that is, a system $A$ of
dimension $d_A \leq d_{A'}$ and a pure state $\rho_{AA'}$ such that
$\tr_{A}(\rho_{AA'}) = \rho_A'$. Every state has such a purification, but it is
not unique \cite{NC00}. Here, for every state $\rho_{A'}$, we consider a
purification with $d_A = d_{A'}$, which is called the \emph{canonical
purification} $\psi_{AA'}^\rho$, and is given by $\psi_{AA'}^\rho =
\ketbra{\psi^\rho}_{AA'}$, where
\begin{align}
  \label{eq:psi}
  \ket{\psi^\rho}_{AA'} = d_{A'} \  \left( \idchan{A} \otimes \sqrt{\rho_{A'}}
  \right) \ket{\Phi}_{AA'} \,.
\end{align}
Here, $\ket{\Phi_{AA'}}$ is the maximally entangled basis with respect to some
bases $\{ \ket{i}_A \}_i$ and $\{ \ket{i}_{A'} \}_i$ for $\ca{H}_A$ and
$\ca{H}_{A'}$, respectively (their choice is irrelevant for what we consider),
\begin{align}
  \label{eq:phi}
  \ket{\Phi}_{AA'} = \sum_{i=1}^{d_{A'}} \frac{1}{\sqrt{d_{A'}}} \ket{i}_A
  \otimes \ket{i}_{A'} \,.
\end{align}
By extending system $A'$ to system $AA'$ in this way, we arrive at the overall
picture shown in Figure~\ref{fig:fully-purified}. After the channel $\chan$ acted on
system $A'$, we not only consider the output system $B$ but the tripartite
system $ABE$, which is in a state $\rho_{ABE}$. Since isometries map pure input
states to pure output states, is is a pure state. This is the reason for this
extension. It will allow us to apply the duality relation in
Section~\ref{app:boundonhmin} (see Lemma~\ref{lem:duality}).

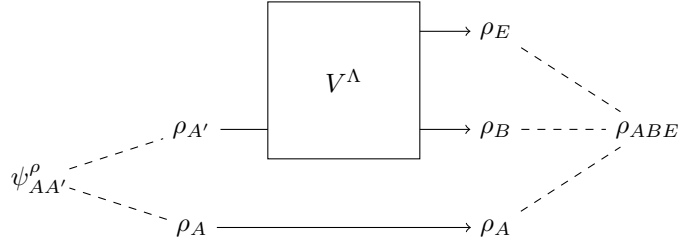
\begin{figure}[h!]
  \centering
  \begin{tikzpicture}[yscale=1.3]
    \node (a) {$\rho_A$};
    \node at ($(a) + (0,1)$) (a') {$\rho_{A'}$};
    \node[inner sep=0] at ($(a) + (-2,0.5)$) (psi) {$\psi_{AA'}^\rho$};
    \draw[dashed] (a) -- (psi) -- (a');

    \node at ($(a) + (4,0)$) (a2) {$\rho_A$};
    \node at ($(a') + (4,0)$) (b) {$\rho_B$};
    \node at ($(b) + (0,1)$) (e) {$\rho_E$};
    \draw[->] (a') -- (b);
    \draw[->] (a) -- (a2);
    \draw[->] ($(e) + (-2,0)$) -- (e);

    \draw[fill=white] ($(a') + (1,-0.3)$) rectangle ($(e) + (-1,0.3)$);
    \node at ($(a') + (2,0.5)$) {$V^\chan$};

    \node at ($(b) + (2,0)$) (abe) {$\rho_{ABE}$};
    \draw[dashed] (e) -- (abe);
    \draw[dashed] (b) -- (abe);
    \draw[dashed] (a2) -- (abe);
  \end{tikzpicture}
  \caption{\textbf{The fully purified diagram for a quantum channel.}
    Since the input state $\psi_{AA'}^\rho$ is pure and $\idchan{A} \otimes
    V^\chan$ is an isometry, the output state $\rho_{ABE}$ is pure.
    \label{fig:fully-purified}}
\end{figure}

Now we are ready to proceed with the next step in the derivation of inequality
\eqref{eq:minentbound}. It turns out that the one-shot capacity of entanglement
transmission $\Qent^\epsilon(\chan)$ of the channel $\chan$ can be bounded by
functions of the state~$\rho_{ABE}$ that we described above. Buscemi and Datta
\cite{BD10} have shown that a lower bound on the one-shot capacity of
entanglement transmission can be formulated in terms of a maximization of an
entropic quantity. Subsequently, Morgan and Winter tightened this bound and
translated it to an optimization of the smooth min-entropy of the state
$\rho_{AE} = \tr_B(\rho_{ABE})$ \cite{MW14}. Here we use this bound in the
following form \cite{TBR16}:
\begin{align}
  \label{eq:qentbound}
  \Forall \epsilon > 0: \quad \Q^\epsilon_{\text{ent}}(\chan) \geq \sup_{\eta
    \in (0,\sqrt{\epsilon})} \sup_{\rho_{A'} \in \ca{S}(\ca{H}_{A'})} \left(
    \Hmin^{\sqrt{\epsilon} - \eta}(A|E)_\rho - 4 \log \frac{1}{\eta} - 1
    \right) \,.
\end{align}
The square root in the smoothing parameter is a consequence of the fact that the
bound \eqref{eq:qentbound} was derived through conversion from a bound where the
figure of merit for entanglement transmission was the purified distance
\cite{MW14} instead of the fidelity \cite{TBR16}.


Next, we drop the maximization over the state $\rho_{A'} \in
\ca{S}(\ca{H}_{A'})$ by choosing the maximally mixed state $\rho_{A'} =
\idop{A'} / d_{A'}$. This way, we arrive at another lower bound:
\begin{align}
  \label{eq:qentbound2}
  \Forall \epsilon > 0: \quad \Q^\epsilon_{\text{ent}}(\chan) \geq \sup_{\eta
    \in (0,\sqrt{\epsilon})} \left( \Hmin^{\sqrt{\epsilon} - \eta}(A|E)_\rho - 4
    \log \frac{1}{\eta} - 1 \right) \,.
\end{align}
This corresponds to the case where the input state $\psi^\rho_{AA'}$ in
Figure~\ref{fig:fully-purified} is given by the maximally entangled state $\Phi_{AA'}$
(see equations \eqref{eq:psi} and \eqref{eq:phi} above). This is of particular
importance for us because we can actually estimate $\Hmin^\epsilon(A|E)$ in that
case (see Section~\ref{app:final}). Combining inequalities \eqref{eq:capbound}
and \eqref{eq:qentbound2}, we get the bound
\begin{align}
  \Forall \epsilon > 0: \quad \Q^\epsilon(\chan) \geq \sup_{\eta \in (0,
    \sqrt{\epsilon / 2})} \left( \Hmin^{\sqrt{\epsilon / 2} - \eta}(A|E)_\rho -
    4 \log \frac{1}{\eta} - 2 \right) \,.
\end{align}

\section{Result: Proof of the bound on the min-entropy}
\label{app:boundonhmin}

In Section~\ref{app:bound-on-capacity}, we have seen that the one-shot quantum capacity
$\Q^\epsilon(\chan)$ can be bounded in terms of the smooth min-entropy
$\Hmin^\epsilon(A|E)$ of an appropriately defined state $\rho_{AE}$. In this
appendix, we prove that this min-entropy is bounded by the smooth max-entropies
$\Hmax^{\epsilon}(X|B)$ and $\Hmax^{\epsilon}(Z|B)$ of measurement $X$ and $Z$
on $A$. More precisely, we will show:
\begin{align}
  \label{eq:main-result}
  \Forall \epsilon > 0, \Forall \epsilon', \epsilon'' \geq 0: \quad
  H_{\text{min}}^{3\epsilon + \epsilon' + 4\epsilon''}(A|E)_\rho \geq Nq
    - \left( H_{\text{max}}^{\epsilon''}(X|B)_\rho
    + H_{\text{max}}^{\epsilon'}(Z|B)_\rho \right)
    - 2\log\frac{2}{\epsilon^2}
\end{align}
(see Theorem~\ref{thm:hminae} below). Using this inequality, we will prove our bound on
the one-shot quantum capacity in terms of the protocol parameters in
Section~\ref{app:final}.

In the following, we will cite some lemmas that we will need for the proof of
the bound \eqref{eq:main-result}. The most important ones are:
\begin{itemize}
  \item an uncertainty relation for the smooth min- and max-entropies
    \cite{TR11},
  \item a chain rule theorem for the smooth max-entropy \cite{VDTR12} and
  \item a duality relation for the smooth min- and max-entropies
    \cite{KRS09,TCR10}.
\end{itemize}
We start with the uncertainty relation for the smooth min- and max-entropies.

\begin{lem}[Smooth min-max uncertainty] \label{lem:ur}
  Let $\rho_{ABE} \in \ca{S}(\ca{H}_{ABE})$ be a pure tripartite state where $A$
  is an $N$-qubit system, let $\xb = \{ \xb_0, \xb_1 \}$ and $\zb = \{ \zb_0,
  \zb_1 \}$ be qubit POVMs. Consider the states
  $\rho_{XBE}$ and $\rho_{ZBE}$ that arise from measuring all of the $N$ qubits
  of system $A$ with respect to $\xb$ and $\zb$, respectively, and storing the
  outcomes in a classical register $X$ and $Z$, respectively,
  \begin{align}
    \rho_{XBE} &= \sum_{x \in \{0,1\}^n} P_X(x) \ \ketbra{x} \otimes \rho_{BE}^x
      \,, \\
    \rho_{ZBE} &= \sum_{z \in \{0,1\}^n} P_Z(z) \ \ketbra{z} \otimes \rho_{BE}^z
    \,.
  \end{align}
  where
  \begin{align}
    \label{eq:prodpovmApp}
    &P_X(x) = \tr \left( \Pi_X(x) \rho_A \right) \,, \\
    &\Pi_X(x) = \bigotimes_{i=1}^n \xb_{x_i} \quad \text{for } x = (x_1, \ldots,
      x_n) \in \{0,1\}^n \,, \\
    &\rho^x_{BE} = \frac{\tr_A((\Pi_X(x) \otimes \idop{BE}) \rho_{ABE} (\Pi_X(x)
      \otimes \idop{BE}))}{P_X(x)} \,,
  \end{align}
  and analogously for $\rho_{ZBE}$.
  Then for $\epsilon \geq 0$,
  \begin{align}
    \label{eq:ur}
    \Hmin^\epsilon(X|E)_\rho + \Hmax^\epsilon(Z|B)_\rho \geq Nq \,,
  \end{align}
  where
  \begin{align}
    \label{eq:prepqual}
    q = - \log \max_{i, j} \left\Vert \sqrt{\xb_i} \sqrt{\zb_j}
    \right\Vert_\infty^2 \,.
  \end{align}
  The parameter $q$ is the \dt{preparation quality}. If $\xb$ and $\zb$ are
  measurements with respect to mutually unbiased bases, then $q=1$.
\end{lem}

The chain rule that we will use is actually just one out of a series of chain
rule inequalities proved in \cite{VDTR12}. The particular form that we use here
can be found in \cite{Tom12}.

\begin{lem}[Chain rule for smooth max-entropy] \label{lem:chainrule}
  Let $\rho_{ABC} \in \ca{S}^{\leq}(\ca{H}_{ABC})$ be a tripartite state, let
  $\epsilon > 0$, $\epsilon' \geq 0$, $\epsilon'' \geq 0$. Then
  \begin{align}
    H^{\epsilon + \epsilon' + 2\epsilon''}_{\text{max}}(AB|C)_\rho \leq
    H^{\epsilon'}_{\text{max}}(A|BC)_\rho +
    H^{\epsilon''}_{\text{max}}(B|C)_\rho + \log\frac{2}{\epsilon^2} \,.
  \end{align}
\end{lem}

The duality relation between the smooth min- and max-entropy, or \emph{min-max
duality}, for short, relates the smooth min-entropy of a state to the
max-entropy of a purification of the state. It was first proved for the
unsmoothed min- and max-entropy K\"onig, Renner and Schaffner in \cite{KRS09}.
The min-max duality for the smooth entropies is due to Tomamichel, Colbeck and
Renner \cite{TCR10}.

\begin{lem}[Min-max duality] \label{lem:duality}
  Let $\rho_{ABE} \in \ca{S}(\ca{H}_{ABE})$ be a pure tripartite state, let
  $\epsilon \geq 0$. Then
  \begin{align}
    \Hmin(A|E)_\rho &= - \Hmax(A|B)_\rho \quad \text{and} \\
    \Hmin^\epsilon(A|E)_\rho &= - \Hmax^\epsilon(A|B)_\rho \,.
  \end{align}
\end{lem}

Apart from these three main ingredients, we will also make use of three smaller
lemmas. The first one states that the smooth min- and max-entropies are
invariant under isometries \cite{Tom12}.

\begin{lem}[Invariance under isometries]
  \label{lem:invariance}
  Let $\rho_{AB} \in \ca{S}^{\leq}(\ca{H}_{AB})$ be a bipartite state, let
  $\epsilon \geq 0$. Then for all isometries $V: \ca{H}_A \rightarrow \ca{H}_{A'}$
  and $W: \ca{H}_B \rightarrow \ca{H}_{B'}$, the embedded state $\sigma_{A'B'} =
  (V \otimes W) \rho_{AB} (V^\dagger \otimes W^\dagger)$ satisfies
  \begin{align}
    \Hmin^\epsilon(A|B)_\rho = \Hmin^\epsilon(A'|B')_\sigma
    \quad \text{and} \quad
    \Hmax^\epsilon(A|B)_\rho = \Hmax^\epsilon(A'|B')_\sigma \,.
  \end{align}
\end{lem}

In simple terms, the following lemma states that ``forgetting'' side information
cannot decrease one's uncertainty. It is a special case of a more general
theorem, called the \emph{data processing inequality} \cite{Tom12}. We only
state the more special case that we are interested in.

\begin{lem} \label{lem:forget}
  Let $\rho_{ABC} \in \ca{S}^{\leq}(\ca{H}_{ABC})$ be a tripartite state. Then
  \begin{align}
    \Hmax(A|BC) \leq \Hmax(A|B) \,.
  \end{align}
\end{lem}

Finally, the last lemma that we add to our list of tools shows how the
(unsmoothed) max-entropy simplifies in the case where classical side information
is given.

\begin{lem}
  \label{lem:class-cond}
  Let $\rho_{ACX} \in \ca{S}^{\leq}(\ca{H}_{ACX})$ be a state of the form
  \begin{align}
    \rho_{ACX} = \sum_x p_x \ \rho_{AC}^x \otimes \ket{x} \bra{x}_X \,, \quad
    \text{where} \quad \rho_{AC}^x \in \ca{S}^{\leq}(\ca{H}_{AC}) \,.
  \end{align}
  Then \cite{Tom12}
  \begin{align}
    \Hmax(A|CX)_\rho = \log\left(\sum_x P_X(x) \
      2^{\Hmax(A|C)_{\rho_{AC}^x}}\right) \,.
  \end{align}
\end{lem}

Now we are ready to state the theorem formally and prove it.

\begin{thm}
  \label{thm:hminae}
  Let $\rho_{ABE} \in \ca{S}(\ca{H}_{ABE})$ be a pure tripartite state where $A$
  and $B$ are each an $N$-qubit system, let $\xb = \{ \xb_0, \xb_1 \}$ and $\zb
  = \{ \zb_0, \zb_1 \}$ be non-trivial projective measurements on a qubit (that
  is, both elements are one-dimensional projectors). Consider the states
  $\rho_{XBE}$ and $\rho_{ZBE}$ that arise from measuring all of the $N$ qubits
  of system $A$ with respect to $\xb$ and $\zb$ (as in Lemma~\ref{lem:ur}).
  Then, for $\epsilon > 0$ and $\epsilon', \epsilon'' \geq 0$, it holds that
  \begin{align}
    \label{eq:rediscool}
    H_{\text{min}}^{3\epsilon + \epsilon' + 4\epsilon''}(A|E)_\rho \geq Nq
      - (H_{\text{max}}^{\epsilon'}(Z|B)_\rho
      + H_{\text{max}}^{\epsilon''}(X|B)_\rho)
      - 2\log\frac{2}{\epsilon^2} \,,
  \end{align}
  where $q$ is the preparation quality (as in Lemma~\ref{lem:ur}).
\end{thm}

\begin{proof}
  Starting from $\rho_{ABE}$, we construct a purification $\rho_{AXX'BE}$ of
  $\rho_{XBE}$. Further below, we will expand the smooth max-entropy of this
  state using the chain rule (Lemma~\ref{lem:chainrule}). Reformulating the terms in
  that expansion will lead us to the desired result.

  Consider the product POVM elements
  \begin{align}
    \label{eq:prodpovm}
    \Pi_X(x) = \bigotimes_{i=1}^N \xb_{x_i} \quad \text{for } x = (x_i)_{i=1}^N
    \in \{0,1\}^N \,.
  \end{align}
  We construct $\rho_{AXX'BE}$ from $\rho_{ABE}$ by performing a coherent
  measurement on the $A$ system with respect to the POVM formed by the elements
  \eqref{eq:prodpovm}. The outcome of this measurement is stored in two copies
  $X$ and $X'$ of a classical register. For $x \in \{0,1\}^N$, let $V_x$ be the
  map
  \begin{align}
    \begin{array}{cccc}
      V_x: & \ca{H}_A    & \rightarrow & \ca{H}_{AXX'} \\
           & \ket{\psi} & \mapsto     & \Pi_X(x) \ket{\psi} \otimes
					\ket{x}_X \otimes \ket{x}_{X'} \,.
    \end{array}
  \end{align}
  We define the state $\rho_{AXX'BE} := V(\rho_{ABE})$, where
  \begin{align}
    \begin{array}{cccl}
      V: & \End(\ca{H}_{ABE}) & \rightarrow & \End(\ca{H}_{AXX'BE}) \\
	 & \rho_{ABE}        & \mapsto     & \sum_x (V_x \otimes \idop{BE})
	 \rho_{ABE} (V_x^\dagger \otimes \idop{BE}) \,.
    \end{array}
  \end{align}
  The map $V$ is an isometry that maps the pure state $\rho_{ABE}$ to the pure
  state $\rho_{AXX'BE}$. Thus, by virtue of Lemma~\ref{lem:duality}, it holds that
  \begin{align}
    \label{eq:touse}
    \Hmin^{\epsilon'}(X|E)_\rho = -\Hmax^{\epsilon'}(X|AX'B)_\rho \,.
  \end{align}
  We will use equation \eqref{eq:touse} further below.

  Now we expand the max-entropy of $\rho_{AXX'BE}$ using the chain rule,
  Lemma~\ref{lem:chainrule}:
  \begin{align}
    \label{eq:chained1}
    \Hmax^{\epsilon + \epsilon' + 2(\epsilon + 2\epsilon'')}(AXX'|B)_\rho \leq
      \Hmax^{\epsilon'}(X|AX'B)_\rho + \Hmax^{\epsilon + 2\epsilon''}
      (AX'|B)_\rho + \log \frac{2}{\epsilon^2} \,.
  \end{align}
  The states $\rho_{AB}$ and $\rho_{AXX'B}$ only differ by an isometry, so by
  Lemma~\ref{lem:invariance}, we have
  \begin{align} \label{eq:replace-a}
    H^{\epsilon + \epsilon' + 2(\epsilon +
      2\epsilon'')}_{\text{max}}(AXX'|B)_\rho  = H^{3\epsilon +
      \epsilon' + 4\epsilon''}_{\text{max}}(A|B)_\rho \,.
  \end{align}
  (It will become clear further below why we choose the smoothing parameter on
  the left hand side this way.) Moreover, the marginals $\rho_{AX'B}$ and
  $\rho_{AXB}$ only differ by a unitary $\ca{H}_X \rightarrow \ca{H}_{X'}$ and
  therefore
  \begin{align}
    \label{eq:marginals-equiv}
    H^{\epsilon + 2\epsilon''}_{\text{max}}(AX'|B)_\rho
    = H^{\epsilon + 2\epsilon''}_{\text{max}}(AX|B)_\rho
  \end{align}
  Combining Equations~\ref{eq:chained1},~\ref{eq:replace-a},and~\ref{eq:marginals-equiv} yields
  \begin{align} \label{combination1}
    \Hmax^{\epsilon'}(X|AX'B)_\rho \geq
    H^{3\epsilon + \epsilon' + 4\epsilon''}_{\text{max}}(A|B)_\rho -
    H^{\epsilon + 2\epsilon''}_{\text{max}}(AX|B)_\rho -
    \log\frac{2}{\epsilon^2} \,.
  \end{align}
  Now we expand the term $\Hmax^{\epsilon + 2\epsilon''}(AX|B)$ using the chain
  rule:
  \begin{align}
    \label{eq:chained2}
    H^{\epsilon + 2\epsilon''}_{\text{max}}(AX|B)_\rho &\leq
      \underbrace{H^0_{\text{max}}(A|XB)_\rho}_{\Hmax(A|XB)_\rho} +
      H^{\epsilon''}_{\text{max}}(X|B)_\rho + \log\frac{2}{\epsilon^2} \,.
  \end{align}
  Combining \eqref{combination1} with \eqref{eq:chained2} allows us to infer
  \begin{align}
    \label{eq:combination2}
    \Hmax^{\epsilon'}(X|A X'B)_\rho \geq H^{3\epsilon
      + \epsilon' + 4\epsilon''}_{\text{max}}(A|B)_\rho - \Hmax(A|XB)_\rho
      - \Hmax^{\epsilon''}(X|B)_\rho - 2\log\frac{2}{\epsilon^2} \,.
  \end{align}
  Now we use equation~\eqref{eq:touse} that we derived above to rewrite
  inequality~\eqref{eq:combination2} as
  \begin{align}
    \Hmin^{\epsilon'}(X|E)_\rho \leq - H^{3\epsilon
      + \epsilon' + 4\epsilon''}_{\text{max}}(A|B)_\rho + \Hmax(A|XB)_\rho
      + \Hmax^{\epsilon''}(X|B)_\rho + 2\log\frac{2}{\epsilon^2} \,.
  \end{align}
  Reordering terms and using Lemma~\ref{lem:forget} and the uncertainty relation for
  the smooth min- and max-entropy (Lemma~\ref{lem:ur}), we get
  \begin{align}
    H^{3\epsilon + \epsilon' + 4\epsilon''}_{\text{max}}(A|B)_\rho
    &\leq \underbrace{\Hmax(A|XB)_\rho}_{\leq \Hmax(A|X)_\rho}
    + \Hmax^{\epsilon''}(X|B)_\rho \underbrace{
      - \Hmin^{\epsilon'}(X|E)_\rho}_{\leq \Hmax^{\epsilon'}(Z|B)_\rho - Nq}
    + 2\log\frac{2}{\epsilon^2} \\
    &\leq \Hmax(A|X)_\rho + \Hmax^{\epsilon''}(X|B)_\rho
    + \Hmax^{\epsilon'}(Z|B)_\rho - Nq + 2\log\frac{2}{\epsilon^2} \,.
    \label{eq:aft-uncert}
  \end{align}
  Applying the duality relation (Lemma~\ref{lem:duality}) to the left hand side of
  Equation~\ref{eq:aft-uncert}, we get
  \begin{align}
    \label{almost-result}
    H^{3\epsilon + \epsilon' + 4\epsilon''}_{\text{min}}(A|E)_\rho \geq Nq
    - \Hmax(A|X)_\rho - \left( \Hmax^{\epsilon''}(X|B)_\rho
    + \Hmax^{\epsilon'}(Z|B)_\rho \right) + 2\log\frac{2}{\epsilon^2} \,.
  \end{align}

  We are left to show that $\Hmax(A|X)_\rho$ is upper bounded by $0$. We
  show, more precisely, that $\Hmax(A|X)_\rho = 0$. This goes as follows.
  \begin{align}
    \rho_{AX} &= \tr_{X'BE}(\rho_{AXX'BE}) \\
    &= \tr_{X'BE} \left( \sum_x (V_x \otimes \idop{BE}) \rho_{ABE}
      (V_x^\dagger \otimes \idop{BE}) \right) \\
    &= \tr_{X'} \left( \sum_x V_x \rho_{A} V_x^\dagger \right) \\
    &= \sum_x \Pi_X(x) \rho_A \Pi_X(x) \otimes \ket{x} \bra{x}_X\\
    &= \sum_x P_X(x) \ \rho_A^x \otimes \ket{x} \bra{x}_X \,,
      \label{classical-form}
  \end{align}
  where
  \begin{align}
      P_X(x) &= \tr(\Pi_X(x) \rho_A) \,, \\
      \rho_A^x &= \frac{\Pi_X(x) \rho_A {\Pi_X(x)}}{P_X(x)} \,.
  \end{align}
  Now we can apply Lemma~\ref{lem:class-cond} to Equation~\ref{classical-form}: By setting the
  system $C$ in the lemma to a trivial system ($\ca{H}_C \simeq \mathbb{C}$), we
  can deduce that
  \begin{align}
    \label{cond-class-eq}
    \Hmax(A|X)_\rho = \log\left( \sum_x P_X(x) \ 2^{\Hmax(A)_{\rho_A^x}}
    \right) \,,
  \end{align}
  where $\Hmax(A)_{\rho_A^x}$ reduces to the unconditional form of the
  max-entropy,
  \begin{align}
    \Hmax(A)_{\rho_A^x} = \log\left\Vert\sqrt{\rho_A^x}\right\Vert_1^2 =
    \log\left(\tr\left(\sqrt{\rho_A^x}\right)\right)^2 \,.
  \end{align}
  Since the $\Pi_X(x)$ are one-dimensional projectors, we have that
  \begin{align}
    \Hmax(A)_{\rho^x_A} = 0 \quad \text{for all } x \in \{0,1\}^N
  \end{align}
  and therefore $\Hmax(A|X)_\rho = 0$, as claimed.
  Thus, we have proved that
  \begin{align}
    H_{\text{min}}^{3\epsilon + \epsilon' + 4\epsilon''}(A|E)_\rho \geq Nq
      - \left( H_{\text{max}}^{\epsilon''}(X|B)_\rho
      + H_{\text{max}}^{\epsilon'}(Z|B)_\rho \right)
      - 2\log\frac{2}{\epsilon^2} \,,
  \end{align}
  which is what we wanted to show.
\end{proof}

\section{Result: Proof of the capacity bound in terms of protocol parameters}
\label{app:final}

\subsection{Comparison to min-entropy estimation in QKD}
\label{app:setup}

In Section~\ref{app:bound-on-capacity}, we have seen that the one-shot quantum capacity
of a channel is bounded by the min-entropy. In the last section, we have seen
how the smooth min-entropy $\Hmin^{3\epsilon + \epsilon' + 4\epsilon''}(A|E)$ can
be bounded in terms of the max-entropy $\Hmax^{\epsilon''}(X|B)$ and
$\Hmax^{\epsilon'}(Z|B)$ of the classical measurement outcomes $X$ and $Z$ on
$A$. This puts us in a very good position, because we already know from quantum
key distribution how to bound these max-entropies: a modern approach to quantum
key distribution based on smooth entropies proves security by bounding exactly
such a quantity.

In that approach, a QKD protocol is devised in which after sifting, Alice and
Bob have $n$ systems where they both measured in the $\xb$-basis and $k$ systems
where they both measured in the $\zb$-basis. Then they exchange their outcomes
in the $Z$-basis and determine the error rate $\ez$. If the error rate $\ez$
does not exceed a specified error tolerance $\etolz$, then they conclude that
\cite{TLGR12,PLWC16,Pfi16}
\begin{align}
  \label{eq:marcobound}
  \Hmax^{\epsilon'}(Z|B)_\rho \leq n h(\etolz + \mu(\epsilon)) \,,
\end{align}
where $h$ denotes the binary entropy function and
\begin{align}
  \label{eq:marcomu}
  \epsilon' = \frac{\epsilon}{\sqrt{\ppass}} \,, \quad
  \mu(\epsilon) = \sqrt{\frac{n+k}{nk} \frac{k+1}{k} \ln \frac{1}{\epsilon}} \,.
\end{align}
Here $\ppass$ is the probability that the correlation test (which checks whether
$\ez \leq \etolz$) is passed, where $p = 1 - \ppass$ is the parameter given in the theorem.
The state $\rho$ in inequality
\eqref{eq:marcobound} is the state of the $n$ qubits that have actually been
measured in $X$. This means that from the error rate $\ez$ in one part of the
qubits, one can infer a bound on $\Hmax^{\epsilon'}(Z|B)$ for the other part of
the qubits. This is illustrated in Figure~\ref{fig:qkdinfer}.

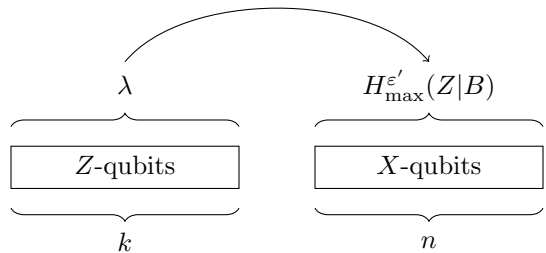
\begin{figure}[h!]
  \centering
  \begin{tikzpicture}[yscale=0.7]
    \node (x) {$X$-qubits};
    \draw ($(x) + (-1.5,-0.4)$) rectangle ($(x) + (1.5 ,0.4)$);
    \node at ($(x) + (-4,0)$) (z) {$Z$-qubits};
    \draw ($(z) + (-1.5,-0.4)$) rectangle ($(z) + (1.5 ,0.4)$);
    \draw[decorate,decoration={brace,amplitude=5pt,mirror,raise=0.05cm}]
      ($(x) + (1.5,0.7)$) -- ($(x) + (-1.5,0.7)$) node[midway,yshift=0.6cm]
      {$\Hmax^{\epsilon'}(Z|B)$};
    \draw[decorate,decoration={brace,amplitude=5pt,mirror,raise=0.05cm}]
      ($(z) + (1.5,0.7)$) -- ($(z) + (-1.5,0.7)$) node[midway,yshift=0.6cm]
      {$\ez$};
    \draw[->] ($(z) + (0,2)$) to[bend left=60] ($(x) + (0,2)$);

    \draw[decorate,decoration={brace,amplitude=5pt,raise=0.05cm}]
      ($(x) + (1.5,-0.7)$) -- ($(x) + (-1.5,-0.7)$) node[midway,yshift=-0.5cm]
      {$n$};
    \draw[decorate,decoration={brace,amplitude=5pt,raise=0.05cm}]
      ($(z) + (1.5,-0.7)$) -- ($(z) + (-1.5,-0.7)$) node[midway,yshift=-0.5cm]
      {$k$};
  \end{tikzpicture}
  \caption{\textbf{Bounding the max-entropy from an error rate on a different
    part.} In the QKD protocol that we consider \cite{TLGR12,PLWC16,Pfi16}, the
    test qubits are measured in $Z$ and the key qubits are measured in $X$. For
    the security of the protocol, $\Hmax^{\epsilon'}(Z|B)$ needs to be bounded
    for the key qubits. This bound can be inferred from the error rate $\ez$ on
    the test qubits.
    \label{fig:qkdinfer}}
\end{figure}

In the QKD scenario we just described, the goal was to infer a bound on
$\Hmax^{\epsilon'}(Z|B)$ on only a part of the total system from the error rate
$\ez$ on its complement. In our one-shot quantum capacity estimation and
verification protocols, the situation is a bit different. It is easier to
discuss the verification protocol first, because it is conceptually closer to
the QKD protocol from which we adopt the estimation techniques.

\subsection{Proof for the verification protocol}
\label{app:verification}

In the verification protocol, the qubits are divided into three subsets: one
subset of test qubits that are measured in the $X$-basis, one subset of test
qubits that are measured in the $Z$-basis, and the data qubits that are not
measured at all (see the top part of Figure~\ref{fig:verification}). In the main
article, we stated the protocol such that each of theses three subets has the
same size $N$.  Here, we consider the more general case where each of theses
subsets might have a different size $n, k, N \in \mathbb{N}_+$, respectively,
and later specialize the result to $n=k=N$. This also helps us in the proof to
keep track of which number we mean. Here, we denote the smoothing parameter by
$\delta$ instead of $\epsilon$ (it will become clear below why it is useful to
do so).

\begin{figure}[h!]
  \centering
  \begin{tikzpicture}[xscale=0.8,yscale=0.7]
    \node (x) {$X$-qubits};
    \node at ($(x) + (-4,0)$) (d) {data qubits};
    \node at ($(d) + (-4,0)$) (z) {$Z$-qubits};
    \draw ($(x) + (-1.5,-0.4)$) rectangle ($(x) + (1.5 ,0.4)$);
    \draw ($(d) + (-1.5,-0.4)$) rectangle ($(d) + (1.5 ,0.4)$);
    \draw ($(z) + (-1.5,-0.4)$) rectangle ($(z) + (1.5 ,0.4)$);
    \draw[decorate,decoration={brace,amplitude=5pt,mirror,raise=0.05cm}]
      ($(d) + (1.5,0.7)$) -- ($(d) + (-1.5,0.7)$) node[midway,yshift=0.6cm]
      {$\Hmax^{\delta'}(Z|B)$, $\Hmax^{\delta''}(X|B)$};
    \draw[decorate,decoration={brace,amplitude=5pt,mirror,raise=0.05cm}]
      ($(z) + (1.5,0.7)$) -- ($(z) + (-1.5,0.7)$) node[midway,yshift=0.6cm]
      {$\ez$};
    \draw[decorate,decoration={brace,amplitude=5pt,mirror,raise=0.05cm}]
      ($(x) + (1.5,0.7)$) -- ($(x) + (-1.5,0.7)$) node[midway,yshift=0.6cm]
      {$\ex$};
    \draw[->] ($(z) + (0,2)$) to[bend left=50] ($(d) + (-1.5,2)$);
    \draw[->] ($(x) + (0,2)$) to[bend right=50] ($(d) + (1.5,2)$);

    \draw[decorate,decoration={brace,amplitude=5pt,raise=0.05cm}]
      ($(x) + (1.5,-0.7)$) -- ($(x) + (-1.5,-0.7)$) node[midway,yshift=-0.5cm]
      {$n$};
    \draw[decorate,decoration={brace,amplitude=5pt,raise=0.05cm}]
      ($(z) + (1.5,-0.7)$) -- ($(z) + (-1.5,-0.7)$) node[midway,yshift=-0.5cm]
      {$k$};
    \draw[decorate,decoration={brace,amplitude=5pt,raise=0.05cm}]
      ($(d) + (1.5,-0.7)$) -- ($(d) + (-1.5,-0.7)$) node[midway,yshift=-0.5cm]
      {$N$};

 \draw[arrows = {->}, line width=0.1cm] ($(x) +
      (0,-2)$) -- ($(x) + (2,-4.5)$);
    \node at ($(x) + (3.5,-3)$) {trace out $Z$-qubits};

    \draw[arrows = {->}, line width=0.1cm] ($(z) +
      (0,-2)$) -- ($(z) + (-2,-4.5)$);
    \node at ($(z) + (-3.5,-3)$) {trace out $X$-qubits};

    \node at ($(z) + (0,-8)$) (d') {data qubits};
    \node at ($(d') + (-4,0)$) (z') {$Z$-qubits};
    \draw ($(d') + (-1.5,-0.4)$) rectangle ($(d') + (1.5 ,0.4)$);
    \draw ($(z') + (-1.5,-0.4)$) rectangle ($(z') + (1.5 ,0.4)$);

    \node (x'') at ($(x) + (4,-8)$) {$X$-qubits};
    \node at ($(x'') + (-4,0)$) (d'') {data qubits};
    \draw ($(d'') + (-1.5,-0.4)$) rectangle ($(d'') + (1.5 ,0.4)$);
    \draw ($(x'') + (-1.5,-0.4)$) rectangle ($(x'') + (1.5 ,0.4)$);

    \draw[decorate,decoration={brace,amplitude=5pt,mirror,raise=0.05cm}]
      ($(x'') + (1.5,0.7)$) -- ($(x'') + (-1.5,0.7)$) node[midway,yshift=0.6cm]
      {$\ex$};
    \draw[decorate,decoration={brace,amplitude=5pt,mirror,raise=0.05cm}]
      ($(d'') + (1.5,0.7)$) -- ($(d'') + (-1.5,0.7)$) node[midway,yshift=0.6cm]
      {$\Hmax^{\delta''}(X|B)$};
    \draw[decorate,decoration={brace,amplitude=5pt,raise=0.05cm}]
      ($(x'') + (1.5,-0.7)$) -- ($(x'') + (-1.5,-0.7)$) node[midway,yshift=-0.5cm]
      {$n$};
    \draw[decorate,decoration={brace,amplitude=5pt,raise=0.05cm}]
      ($(d'') + (1.5,-0.7)$) -- ($(d'') + (-1.5,-0.7)$)
      node[midway,yshift=-0.5cm] {$N$};
    \draw[->] ($(x'') + (0,2)$) to[bend right=50] ($(d'') + (0,2)$);

    \draw[decorate,decoration={brace,amplitude=5pt,raise=0.05cm}]
      ($(z') + (1.5,-0.7)$) -- ($(z') + (-1.5,-0.7)$) node[midway,yshift=-0.5cm]
      {$k$};
    \draw[decorate,decoration={brace,amplitude=5pt,raise=0.05cm}]
      ($(d') + (1.5,-0.7)$) -- ($(d') + (-1.5,-0.7)$)
      node[midway,yshift=-0.5cm] {$N$};
    \draw[decorate,decoration={brace,amplitude=5pt,mirror,raise=0.05cm}]
      ($(d') + (1.5,0.7)$) -- ($(d') + (-1.5,0.7)$) node[midway,yshift=0.6cm]
      {$\Hmax^{\delta'}(Z|B)$};
    \draw[decorate,decoration={brace,amplitude=5pt,mirror,raise=0.05cm}]
      ($(z') + (1.5,0.7)$) -- ($(z') + (-1.5,0.7)$) node[midway,yshift=0.6cm]
      {$\ez$};
    \draw[->] ($(z') + (0,2)$) to[bend left=50] ($(d') + (0,2)$);
  \end{tikzpicture}
  \caption{\textbf{Inference of the max-entropies in the verification protocol.}
    Our verification protocol can be seen as the parallel execution of two QKD
    estimation protocols (see Figure~\ref{fig:qkdinfer}). Thus, if both $\ex \leq
    \etolx$ and $\ez \leq \etolz$, we get two bounds of the form
    \eqref{eq:marcobound}.
    \label{fig:verification}}
\end{figure}
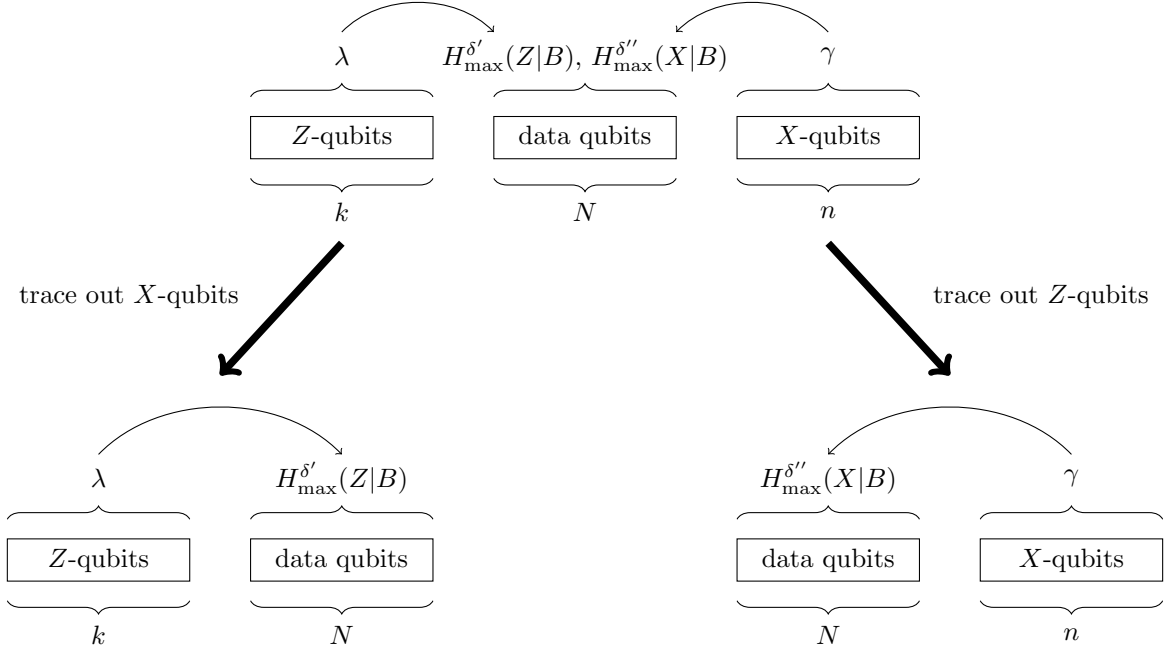

This situation may look more complicated than in the QKD scenario above.
However, it turns out that our verification protocol can be seen as running the
above QKD estimation two times in parallel (see the bottom part of
Figure~\ref{fig:verification}). When we trace out the $X$-qubits, the remainder is in
the same situation as in the QKD case as shown in Figure~\ref{fig:qkdinfer}, with the
$N$ data qubits taking the role of the qubits for which we bound the max-entropy
$\Hmax^{\delta'}(Z|B)$. Therefore, if we find that the error rate $\ez$ is below
a tolerated error rate $\etolz$, we conclude that
\begin{align}
  \Hmax^{\delta'}(Z|B) \leq N h(\etolz + \muz(\delta)) \,,
\end{align}
where
\begin{align}
  \label{eq:muz}
  \delta' = \frac{\delta}{\sqrt{\ppass^z}} \,, \quad
  \muz(\delta) = \sqrt{\frac{N+k}{Nk} \frac{k+1}{k} \ln \frac{1}{\delta}}
\end{align}
and where $\ppass^z$ is the probability that $\ez \leq \etolz$.

Likewise, when we trace out the $Z$-qubits, the remainder looks like in the QKD
case, with the $X$-basis taking the role of the $Z$-basis and with the data
qubits taking the role of the qubits for which we bound the max-entropy
$\Hmax^{\delta''}(X|B)$. If we find that the error rate $\ex$ is below a
tolerated error rate $\etolx$, then
\begin{align}
  \Hmax^{\delta''}(X|B) \leq N h(\etolx + \mux(\delta)) \,,
\end{align}
where
\begin{align}
  \label{eq:mux}
  \delta'' = \frac{\delta}{\sqrt{\ppass^x}} \,, \quad
  \mux(\delta) = \sqrt{\frac{N+n}{Nn} \frac{n+1}{n} \ln \frac{1}{\delta}}
\end{align}
and where $\ppass^x$ is the probability that $\ex \leq \etolx$.

According to our verification protocol (see Protocol~\ref{prot:verification} in the main
article), we are interested in the case where both $\ex \leq \etolx$ and $\ez
\leq \etolz$. In that case, we can conclude that
\begin{align}
  \label{eq:verifbound}
  \Hmax^{\delta'}(Z|B) + \Hmax^{\delta''}(X|B) \leq
  N \Big( h(\etolz + \muz(\delta)) + h(\etolx + \mux(\delta)) \Big) \,.
\end{align}
At this point, we can connect this bound with the bound that we derived in
Section~\ref{app:boundonhmin} (see Theorem~\ref{thm:hminae}), which says that
\begin{align}
  \label{eq:hminaedelta}
  H_{\text{min}}^{3\delta + \delta' + 4\delta''}(A|E)_\rho \geq Nq
    - (H_{\text{max}}^{\delta'}(Z|B)_\rho
    + H_{\text{max}}^{\delta''}(X|B)_\rho)
    - 2\log\frac{2}{\delta^2} \,.
\end{align}
Combining inequalities \eqref{eq:verifbound} and \eqref{eq:hminaedelta}, we get
that
\begin{align}
  \label{eq:comb1}
  H_{\text{min}}^{3\delta + \delta' + 4\delta''}(A|E)_\rho \geq
  N \Big( q  - h(\etolz + \muz(\delta)) + h(\etolx + \mux(\delta)) \Big)
    - 2\log\frac{2}{\delta^2} \,.
\end{align}
This, in turn, can be connected with the min-entropy bound on the one-shot
quantum capacity that we recapitulated in Section~\ref{app:bound-on-capacity}, which
reads
\begin{align}
  \label{eq:caboundrecited}
  \Forall \epsilon > 0: \quad \Q^\epsilon(\chan) \geq \sup_{\eta \in (0,
    \sqrt{\epsilon / 2})} \left( \Hmin^{\sqrt{\epsilon / 2} - \eta}(A|E)_\rho -
    4 \log \frac{1}{\eta} - 2 \right) \,.
\end{align}
To connect inequalities \eqref{eq:comb1} and \eqref{eq:caboundrecited}, we make
a variable transformation such that
\begin{align}
  3\delta + \delta' + 4\delta'' = \sqrt{\epsilon / 2} - \eta \,,
\end{align}
where
\begin{align}
  3\delta + \delta' + 4\delta'' = \delta
  \left( 3 + \frac{1}{\sqrt{\ppass^z}} + \frac{4}{\sqrt{\ppass^x}} \right) \,.
\end{align}
Hence, we get
\begin{align}
  \label{eq:verifgeneral}
  \Forall \epsilon > 0: \quad \Q^\epsilon(\chan) \geq \sup_{\eta \in (0,
    \sqrt{\epsilon / 2})} N \Big( q  - h(\etolz + \muz(\delta)) + h(\etolx +
    \mux(\delta)) \Big) - 2\log\frac{2}{\delta^2} - 4 \log \frac{1}{\eta} - 2
\end{align}
with
\begin{align}
  \label{eq:deltageneral}
  \delta = \frac{ \sqrt{\epsilon / 2} - \eta }{3 + \frac{1}{\sqrt{\ppass^z}} +
    \frac{4}{\sqrt{\ppass^x}}}
\end{align}
and with $\muz$ and $\mux$ as in \eqref{eq:muz} and \eqref{eq:mux},
respectively. This is the general version of our bound for the verification
protocol.

To derive the form of the bound that we presented in the main article, we make
two simplifications. Firstly, we consider the probability that \emph{both} $\ex
\leq \etolx$ \emph{and} $\ez \leq \etolz$, and denothe this joint probability by
$\ppass$. It bounds both probabilities from below, i.e. $\ppass \leq \ppass^z$
and $\ppass \leq \ppass^x$. We set $\pabort = 1 - \ppass$. Thus, the bound
\eqref{eq:verifgeneral} also holds with
\begin{align}
  \label{eq:insert1}
  \delta = \frac{ \sqrt{\epsilon / 2} - \eta }{3 + \frac{5}{\sqrt{1-\pabort}}}
  \,.
\end{align}
Secondly, we set $k=n=N$, and get
\begin{align}
  \label{eq:insert2}
  \muz(\delta) = \mux(\delta) = \mu(\delta) = \sqrt{ \frac{2(N+1)}{N^2}
    \ln \frac{1}{\delta} } \,.
\end{align}
Inserting equations \eqref{eq:insert1} and \eqref{eq:insert2} into equation
\eqref{eq:verifgeneral} gives us the form that we used in the main article.

\subsection{Proof for the estimation protocol}
\label{app:estimation}

Our estimation protocol has one essential difference to the verification
protocol. In the verification protocol, the two error rates $\ex$ and $\ez$ that
are measured do not enter the bound directly. Instead, they are compared with
some maximally tolerated error rates $\etolx$ and $\etolz$, and the bound is a
function of these values. In the estimation protocol, there are no preset
maximally tolerated error rates.  Alice and Bob simply measure two error rates
$\exest$ and $\ezest$, and the bound that they use is a function of these
measured error rates. This may seem different to the verification protocol, but
using a simple argument, we can see that the situation in the estimation
protocol is analogous to the situation in the verification protocol. (This will
also explain why we use the same notation as for the preset values $\etolx$ and
$\etolz$).

Suppose that Alice and Bob run the estimation protocol (see
Protocol~\ref{prot:estimation} of the main article) up to the point where they determine
the error rates. For now, let us denote these error rates by $\ex$ and $\ez$.
Imagine that at that point, Alice and Bob decide that they actually wanted to
make a test in which they check whether both $\ex \leq \etolx$ and $\ez \leq
\etolz$ holds. However, in contrast to the verification protocol, where $\etolx$
and $\etolz$ are preset values, Alice and Bob say that they simply want to make
the test for values of $\etolx$ and $\etolz$ that are exactly equal to the the
error rates that they have just measured, $\etolx = \ex$ and $\etolz = \ez$.
Obviously, if Alice and Bob design the test in this way, they will always pass
the test. Moreover, the interpretation of the passing probability changes: it is
no longer the probability that the measured error rates are below some preset
values. Instead, it becomes the probability that in any run, the measured error
rates stay below the rates that have been measured in this run (more precisely,
it is a lower bound on it). This probability can be seen as a measure for the
typicality of the protocol run, so we may denote it by
$p_{\textnormal{typical}}$. In Protocol~\ref{prot:estimation} in the main article, we
use the complementary probability $p = 1 - p_{\textnormal{typical}}$. Using this
argument, we can consider the state conditioned on passing a correlation test,
just as in the case of the verification protocol.

For this reason, the same general form of the bound \eqref{eq:verifgeneral}
with the same function $\delta$ as in equation \eqref{eq:deltageneral} holds as
for the verification protocol. However, in the estimation protocol, we use the
measured error rates to infer the max-entropies $\Hmax^{\epsilon'}(Z|B)$ and
$\Hmax^{\epsilon''}(X|B)$ for \emph{all} qubits, rather than just on a part that
has not been measured. Therefore, the functions $\mux$ and $\muz$ differ from
the functions for the verification protocol. In order to derive the form of
these functions, we again consider a slight generalization of the protocol that
we considered in the main article. In the main article, we assumed that the $N$
qubits that go through the channel are divided into $N/2$ qubits that are
prepared and measured in the $X$-basis and $N/2$ qubits that are measured in the
$Z$-basis. Here we assume that $n$ qubits are measured in $X$ and $k$ qubits are
measured in $Z$, with $n+k=N$. We denote the measured error rate in $X$ by $\ex$
and the measured error rate in $Z$ by $\ez$. This is shown in
Figure~\ref{fig:estimation}. In order to bound $\Hmax^{\epsilon'}(Z|B)$ from $\ez$, we
follow the original derivation of equation \eqref{eq:marcomu} as in reference
\cite{TLGR12}, adjusted to the situation shown in Figure~\ref{fig:estimation}. For a
detailed derivation, see also \cite{PLWC16,Pfi16}.

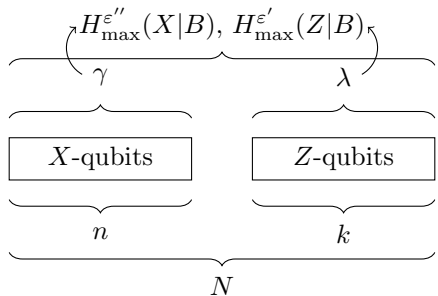
\begin{figure}[h!]
  \centering
  \begin{tikzpicture}[xscale=0.8,yscale=0.7]
    \node (x) {$X$-qubits};
    \draw ($(x) + (-1.5,-0.4)$) rectangle ($(x) + (1.5 ,0.4)$);
    \node at ($(x) + (4,0)$) (z) {$Z$-qubits};
    \draw ($(z) + (-1.5,-0.4)$) rectangle ($(z) + (1.5 ,0.4)$);
    \draw[decorate,decoration={brace,amplitude=5pt,mirror,raise=0.05cm}]
      ($(x) + (1.5,0.7)$) -- ($(x) + (-1.5,0.7)$) node[midway,yshift=0.6cm]
      {$\ex$};
    \draw[decorate,decoration={brace,amplitude=5pt,mirror,raise=0.05cm}]
      ($(z) + (1.5,0.7)$) -- ($(z) + (-1.5,0.7)$) node[midway,yshift=0.6cm]
      {$\ez$};
    \draw[decorate,decoration={brace,amplitude=5pt,mirror,raise=0.05cm}]
      ($(z) + (1.5,1.7)$) -- ($(x) + (-1.5,1.7)$) node[midway,yshift=0.6cm]
      {$\Hmax^{\epsilon''}(X|B)$, $\Hmax^{\epsilon'}(Z|B)$};
    \draw[->] ($(z) + (0.3,1.6)$) to[bend right=50] ($(z) + (0.4,2.5)$);
    \draw[->] ($(x) + (-0.3,1.6)$) to[bend left=50] ($(x) + (-0.4,2.5)$);

    \draw[decorate,decoration={brace,amplitude=5pt,raise=0.05cm}]
      ($(x) + (1.5,-0.7)$) -- ($(x) + (-1.5,-0.7)$) node[midway,yshift=-0.5cm]
      {$n$};
    \draw[decorate,decoration={brace,amplitude=5pt,raise=0.05cm}]
      ($(z) + (1.5,-0.7)$) -- ($(z) + (-1.5,-0.7)$) node[midway,yshift=-0.5cm]
      {$k$};
    \draw[decorate,decoration={brace,amplitude=5pt,raise=0.05cm}]
      ($(z) + (1.5,-1.7)$) -- ($(x) + (-1.5,-1.7)$) node[midway,yshift=-0.5cm]
      {$N$};
  \end{tikzpicture}
  \caption{\textbf{Inference of the max-entropies in the estimation protocol.}
    In the estimation protocol, the measured error rates are used to bound the
    max-entropies on the total system of \emph{all} qubits. This is in contrast
    to the verification protocol, where the measured error rates were used to
    bound the max-entropies on only a part of the total system. This is why we
    cannot simply use the function $\mu$ as in equation \eqref{eq:marcomu} but
    need to derive them for this particular situation.
    \label{fig:estimation}}
\end{figure}

We consider the \emph{Gedankenexperiment} in which all the bits have been
measured in the $Z$-basis. We denote random variable of the error rate $\ez$ in
$Z$ in the $Z$-bits by $\Ez = \Ez_z$, the error rate in $Z$ in the $X$-bits by
$\Ez_x$ and the total error rate in $Z$ by $\Ez_\textnormal{tot}$. Then it holds
that
\begin{align}
  \label{eq:lambdasum}
  \nx \Ez_x + \nz \Ez_z = (\nx + \nz) \Ez_\textnormal{tot} \,.
\end{align}
The division of the qubits into $X$-qubits and $Z$-qubits is fully random.
Therefore, the error number probabilities follow a hypergeometric distribution.
This means that Serfling's bound \cite{Ser74} applies. Here, we use the
particular form presented in inequality (1.3) in \cite{GW15}.
\begin{align}
  \label{eq:serfling}
  \Forall \nu > 0: \quad P[\sqrt{\nx} (\Ez_x - \Ez_\textnormal{tot}) \geq \nu]
  \leq \exp \left( -2\nu^2 \frac{1}{1-\frac{\nx -1}{\nx+\nz}} \right)
\end{align}
Using \eqref{eq:lambdasum}, it is easy to show that
\begin{align}
  \sqrt{\nx} (\Ez_x - \Ez_\textnormal{tot}) \geq \nu \iff
  \Ez_\textnormal{tot} \geq \Ez_z + \frac{\sqrt{\nx}}{\nz} \nu \,.
\end{align}
Therefore, \eqref{eq:serfling} is equivalent to
\begin{align}
  \Forall \nu > 0: \quad P\left[ \Ez_\textnormal{tot} \geq \Ez_z +
    \frac{\sqrt{\nx}}{\nz} \nu \right] \leq \exp \left( -2\nu^2
    \frac{1}{1-\frac{\nx -1}{\nx+\nz}} \right)
\end{align}
With the variable substitution
\begin{align}
  \nu = \frac{k}{\sqrt{n}} \mu_z \,, \quad \text{so that} \quad
  \mu_z = \frac{\sqrt{n}}{k} \nu \,,
\end{align}
we can write this as
\begin{align}
  \Forall \mu_z > 0: \quad P \left[ \Ez_\textnormal{tot} \geq \Ez_z + \mu_z \right]
    & \leq \exp \left(-2 \left(\frac{k}{\sqrt{n}} \mu_z \right)^2
      \frac{1}{1-\frac{\nx -1}{\nx+\nz}} \right) \\
    & = \exp \left( -2 \frac{\nz^2(\nx+\nz)}{\nx(\nz+1)} \mu_z^2 \right) \,.
\end{align}
According to Bayes' theorem, it holds that
\begin{align}
  P \left[ \Ez_\textnormal{tot} \geq \Ez_z + \mu_z \middle\vert \Ez_z \leq \etolz
  \right] \leq \frac{ P[\Ez_\textnormal{tot} \geq \Ez_z + \mu_z] }
  {P[\Ez_z \leq \etolz]}
\end{align}
and thus
\begin{align}
  \label{eq:estimationtail}
  P \left[ \Ez_\textnormal{tot} \geq \Ez_z + \mu_z \middle\vert \Ez_z \leq \etolz
  \right] \leq \frac{\epsilon^2}{\ppass^z} \,,
\end{align}
where
\begin{align}
  \label{eq:estimationepsilon}
  &\epsilon = \exp \left( - \frac{\nz^2(\nx+\nz)}{\nx(\nz+1)} \mu_z^2 \right)
    \,, \\
  &\ppass^z = P[\Ez_z \leq \etolz] \,.
\end{align}
In \cite{TLGR12,Pfi16}, it was shown that inequality \eqref{eq:estimationtail}
implies that the state of the total system of $\N=\nx+\nz$ qubits, conditioned
on $\Ez_z \leq \etolz$, satisfies
\begin{align}
  \Hmax^{\epsilon / \sqrt{\ppass^z}}(Z|B)_\rho \leq \N h(\etolz + \mu_z) \,,
\end{align}
where $\mu_z$ is solved for in \eqref{eq:estimationepsilon},
\begin{align}
  \mu_z = \sqrt{ \frac{\nx(\nz+1)}{\nz^2(\nx+\nz)} \ln \frac{1}{\epsilon} } \,.
\end{align}

The derivation of $\mux$ is essentially analogous: The only difference is that $n$
and $k$ change their roles, and that $\ppass^x$ replaces $\ppass^z$. Thus, we
get
\begin{align}
  \Hmax^{\epsilon / \sqrt{\ppass^x}}(X|B)_\rho \leq \N h(\etolx + \mu_x) \,,
\end{align}
where
\begin{align}
  &\mu_x = \sqrt{ \frac{\nz(\nx+1)}{\nx^2(\nx+\nz)} \ln \frac{1}{\epsilon} } \,,
    \\
  &\ppass^x = P[\Gamma_x \leq \etolx] \,.
\end{align}

To get the form that we use in the main article, we make again make two
simplifications. As for the verification protocol, we bound $\ppass^x$ and
$\ppass^z$ by a joint passing probability $p_\textnormal{typical}$ and set $p =
1 - p_\textnormal{typical}$.  Finally, we set $n = k = N/2$ and get
\begin{align}
  \mux = \muz = \mu = \sqrt{ \frac{N+2}{N^2} \ln \left( \frac{
    3+\frac{5}{\sqrt{1-p}} }{ \sqrt{\epsilon / 2} - \eta} \right) } \,.
\end{align}
This completes the proof.

\end{document}